\documentclass[11pt]{elsarticle}
\usepackage[a4paper,margin=2cm]{geometry}

\usepackage[utf8]{inputenc}
\usepackage{color}
\usepackage{float}
\usepackage{amsmath}
\usepackage{amssymb}
\usepackage{esint}
\usepackage{hyperref}
\usepackage{varioref}
 \usepackage{amsthm}

\newtheorem{theorem}{Theorem}
\newtheorem{corollary}{Corollary}
\newtheorem{lemma}{Lemma}

\newtheorem{remark}{Remark}
\newtheorem{proposition}{Proposition}

\usepackage{algorithm}
\usepackage{algpseudocode}
\usepackage{algorithmicx}
\let\OldStatex\Statex
\renewcommand{\Statex}[1][3]{%
  \setlength\@tempdima{\algorithmicindent}%
  \OldStatex\hskip\dimexpr#1\@tempdima\relax}

\def\fn{\mathbb{F}^n}

\def\ff{\mathbb{F}}

\def\al{\alpha}

\def\be{\beta}

\def\fn{\mathbb{F}^n}

\def\fq{\mathbb{F}_q}
\def\fqn{\mathbb{F}_q^n}

\def\mcb{\mathcal{B}}
\def\beq{\begin{equation}}
\def\eeq{\end{equation}}

\def\Ffq{\mathcal{F}(\mathbb{F}_q)}

\def\Ffn{\mathcal{F}({\mathbb{F}^n})}
\def\Funfq{\mathcal{F}(\mathbb{F}_q)}

\def\ffnz{\mathcal{F}(\mathbb{F}^n)}
\def\Fl{F_{\lambda}}
\def\Ml{M_{\lambda}}
\def\mcR{\mathcal{R}}
\def\hsl{\hat{\psi}_\lambda}
\def\Gl{G_\lambda}
\def\mM{\mathcal{M}}
\def\fb{f_\beta}
\def\be{\beta}

\def\zq{\mathbb{Z}_q}
\def\modq{\ mod\ q-1}
\def\zqm{\mathbb{Z}_{q-1}}
\def\mckf{\mathbf{K}}
\def\mckfl{\mathbf{K}_{\lambda}}
\def\CSf{\Sigma_f}

\journal{Finite Fields and their Applications}

\begin{document}
\begin{frontmatter}
 \title{On Linear Representation, Complexity and Inversion of Maps over Finite Fields}

\author[1]{Ramachandran Ananthraman\corref{cor1}}
\ead{ramachandran.chittur@unamur.be}

\author[2]{Virendra Sule}
\ead{vrs@ee.iitb.ac.in}


\affiliation[1]{organization={Department of Mathematics and Namur Institute for Complex Systems (naXys), University of Namur},
                city={Namur},
                country={Belgium}}

\affiliation[2]{organization={Department of Electrical Engineering, Indian Institute of Technology Bombay},
                city={Mumbai},
                country={India}}


 \cortext[cor1]{Corresponding author}
\begin{abstract}
This paper defines a linear representation for nonlinear maps $F:\mathbb{F}^n\rightarrow\mathbb{F}^n$ where $\mathbb{F}$ is a finite field, in terms of matrices over $\mathbb{F}$. This linear representation of the map $F$ associates a unique number $N$ and a unique matrix $M$ in $\mathbb{F}^{N\times N}$, called the Linear Complexity and the Linear Representation of $F$ respectively, and shows that the compositional powers $F^{(k)}$ are represented by matrix powers $M^k$. It is shown that for a permutation map $F$ with representation $M$, the inverse map has the linear representation $M^{-1}$. This framework of representation is extended to a parameterized family of maps $F_{\lambda}(x): \mathbb{F} \to \mathbb{F}$, defined in terms of a parameter $\lambda \in \mathbb{F}$, leading to the definition of an analogous linear complexity of the map $F_{\lambda}(x)$, and a parameter-dependent matrix representation $M_\lambda$ defined over the univariate polynomial ring $\mathbb{F}[\lambda]$. Such a representation leads to the construction of a parametric inverse of such maps where the condition for invertibility is expressed through the unimodularity of this matrix representation $M_\lambda$. Apart from computing the compositional inverses of permutation polynomials, this linear representation is also used to compute the cycle structures of the permutation map. Lastly, this representation is extended to a representation of the cyclic group generated by a permutation map $F$, and to the group generated by a finite number of permutation maps over $\mathbb{F}$. 

\end{abstract}

\end{frontmatter}

\section*{Notations}
$\ff$ denotes a finite field and $\ff_q$ denotes the finite field of $q$ elements, where $q = p^m$, for some prime $p$. $\Ffq$ denotes the space of $\fq$-valued functions over $\fq$. $\ff^n$ denotes the Cartesian space of $n$-tuples over $\ff$ and $\Ffn$ denotes the space of $\ff$-valued functions over $\ff^n$. $Id_{\fn}$ denotes the identity map over $\fn$. The coordinate functions $\chi_i \in \Ffn$ are defined as $\chi_i(x) = x_i$ for $x \in \ff^n$ and its $i$-th co-ordinate $x_i \in \ff$. In general, one variable maps are denoted by lower case $f$ and multivariate maps are denoted by upper case $F$.
\section{Introduction}
Given an $n$-dimensional vector space $\ff^n$ over finite field $\ff$, maps $F: \ff^n \to \ff^n$ are ubiquitous in the representation of Finite Automata \cite{Gill2,SeqDyn}, recurrence sequences through Feedback Shift Registers (FSR) \cite{GolombFSR,Goresky}, mathematical models of Stream Ciphers \cite{Rueppel}, state updates of Genetic Networks \cite{HDJ,Kauffman, KO_EK, RThomas} to name a few. In these applications, computations of compositions and inverses of such maps, as well as their representations in polynomials over the finite field $\ff$, play an important role.

This paper defines a linear representation of polynomial maps $F$ over finite fields $\ff$ as matrices $M$ 
over $\ff$ of smallest size $N$. The number $N$ is defined as the \emph{Linear Complexity} of $F$ over $\ff$. This representation is a one-to-one homomorphism of the semigroup generated by $F$ under compositions (group, when $F$ is a permutation) into the semigroup (group) generated by the matrix $M$ under multiplication. This representation also gives a formula for the inverse map $G=F^{-1}$. It facilitates computation of the cycle structure of orbits of the permutation map $F$ in terms of the cycle structure of the linear permutation $M$ over $\ff^N$. This representation is extended for groups generated by multiple (finite) invertible maps $F_i$ over $\ff$ under composition. 

\subsection{Permutation polynomials}
Univariate polynomials $f(x): \ff \to \ff$ that induces a bijection over the field $\ff$ are called \textit{permutation polynomials} (in short, PP) and have been studied extensively in the literature. For instance, given a general polynomial $f(x)$ over $\ff$ deciding whether it is a PP is a well-researched problem in literature \cite{MullenPanario}. Though computational verification of a given polynomial $f(x)$ to be PP is a polynomial time problem in its degree $d$, conditions for any polynomial to be PP is well understood only for certain polynomials with specific structures such as monomials, linearized polynomials, and Dickson polynomials, to name a few. 

Since a permutation polynomial $f(x)$ induces a permutation over the field $\ff$, for every $f(x)$ one can associate another permutation polynomial $g(x)$, defined as the \emph{compositional inverse} of $f$ such that 
\[
(f \circ g) (x) = (g \circ f) (x) = Id_{\ff}.
\]
Computing the compositional inverse of a PP is a hard problem \cite{MullenPanario}. Some work \cite{Cafureetal, CH_2002, TuxanidyWang, ZhengWangWei} study on the compositional inverse, but are mostly restricted and applicable to PPs with special structures.

Over a finite field $\fq$, every polynomial $f(x)$ can be associated with an element $\bar{f}(x)$ of the ring $R[x]$, where $R[x]$ is the polynomial ring with polynomials upto degree $q$. The ring addition is polynomial addition and ring multiplication is polynomial multiplication defined through modulo $x^q-x$, the field polynomial. In general, 
\[
f(x) \sim \bar{f}(x) = f(x)\ \mbox{mod} \ (x^q-x)
\]
and $\bar{f}(x)$ is the \emph{polynomial function} corresponding to the polynomial $f(x)$. 
This ring $R[x]$ is the set of all possible $\fq$-valued functions over $\fq$. The total number of such polynomial functions are $q^{q}$, and they form a vector space of dimension $q$. It can be easily verified that $f(x)$ is a permutation polynomial if and only if its polynomial function $\bar{f}$ is a permutation polynomial. In this manuscript, without loss of generality, we assume the polynomial $f$ over $\fq$ to be of degree less than or equal to $q$ and any polynomial $f$ of higher degree is converted to its equivalent polynomial $\bar{f}$.

\subsubsection{Cycle structures of a permutation polynomial}
Given a permutation polynomial $f(x)$ over the field $\fq$, each point $\alpha \in \fq$ is on a unique orbit (or cycle) under the action of the permutation polynomial and the orbit (or cycle) length of the orbit containing $\alpha$ is the smallest $l \in \mathbb{Z}_+$ such that 
\[
p^{(l)} (\alpha) = \alpha
\]
where 
\[
p^{(l)}  := \underbrace{p \circ p \circ \dots \circ p(x)}_{l\ \text{times}}.
\]
Any permutation polynomial $f(x)$ decomposes the finite field $\fq$ into sets containing mutually exclusive orbits, with the cardinality of each set being equal to the cycle length of the elements in that set. The \emph{cycle structure}\footnote{The cycle structure of a permutation $\CSf$ in general is the set containing information about cycle (or orbit) lengths along with their multiplicities. But in this work, we use this term to only denote the orbit lengths of distinct cycles without considering their multiplicity.} $\CSf$ of a permutation polynomial $f(x)$ is the set of all cycle lengths of the permutation. Computing the cycle structure of permutations represented by PPs is an important problem encountered in cryptography, coding theory, and communication systems \cite{YLC, JL_JB} with no known efficient algorithm for a general class of PPs. Computing the cycle structure of specific forms of PP is a well-studied problem in the theory of finite fields \cite{Ayca,LidlMullen_Dickson,Mullen2016}.

\subsubsection{Parametric family of permutation polynomials}
There has been extensive study about a family of polynomial maps defined through a parameter $a \in \ff$ over finite fields. Some well-studied families of polynomials include the Dickson polynomials and reverse Dickson polynomials, to name a few. Conditions for such families of maps to define a permutation of the field $\ff$ are well studied and established for special classes like Dickson polynomials \cite{LidlBook}, linearized polynomials \cite{BW_ZL} and few other specific forms \cite{TuxanidyWang,ZhengWangWei} to name a few.

In this work, we propose a linear algebraic framework focused on maps over (vector spaces over) finite field with the following objectives 
\begin{itemize}
    \item Given a polynomial function $f(x)$ over a finite field $\ff$ (or $\ff^n $), determine if it is a permutation over $\ff$ ($\ff^n$), and if it is, compute its compositional inverse. 
    \item For a general class of permutations over finite fields $\ff$ (or $\ff^n$), compute the possible cycle lengths. 
    \item Given an $1$-parameter family of maps over $\ff$, determine if it is parametrically invertible over $\ff$. It is also shown in this paper that the compositional inverse of a $1$-parameter family of permutation polynomials is also a $1$-parameter family of permutation polynomials over the same parameter and an explicit construction of the same is given.
\end{itemize}

\subsection{Linear complexity of sequences over finite field} 
The linear complexity of sequences over finite fields has been well studied in cryptography and used to define a measure of pseudorandomness of a sequence \cite{GolombFSR, Rueppel}. Given a sequence $S = \{s_n\}$, the linear complexity is the degree $m$ of the minimal degree polynomial $p(X)$,
\beq
\label{eq:MPseq}
p(X) = \al_0 + \al_1 X + \dots + \al_{m-1}X^{m-1} +  X^{m}
\eeq
which defines the recurrence 
\[
s_{k+m} + \al_{m-1} s_{k+m-1} + \dots + \al_1 s_{k+1} + \al_0 s_{k} = 0
\]
for all $k \geq 0$. This polynomial is the \emph{minimal polynomial} of the sequence $S$. The well-known algorithm by Berlekamp-Massey \cite{Berlekamp, Massey} computes the linear complexity of such sequences over finite fields. This paper extends this notion of linear complexity from that for sequences over finite fields to maps defined over finite fields. In particular, given a map $F: \ff^n \to \ff^n$, we will find a polynomial $p(X)$ as in equation (\ref{eq:MPseq}) with coefficients $\al_i \in \ff$ such that 
\[
F^{(k+N)}(x) + \al_{N-1} F^{(k+N-1)}(x) + \dots + \al_1 F^{(k+1)}(x) + \al_0 F^{(k)}(x) = 0 
\]
In essence, this notion of linear complexity of a function can be used as a characterization for the computational effort involved for computations on $F(x)$ such as computing the cycle structures of the map, computing its compositional inverse etc.,   

\subsection{Connections to dynamical systems over finite field}
\label{ss:DSFF}
An interesting connection is drawn from dynamical systems theory for studying linear complexity and representation of maps over $\ff^n$. Since we are interested in understanding the action of the polynomial map $F(x)$ under its compositions, it is natural to define a discrete-time dynamical system evolving over $\ff^n$ in the following way.
\beq
\label{eq:DS}
x(k+1) = F(x(k))
\eeq
where $k \in \mathbb{Z}_+$ is the time index, $x(k) \in \ff^n$ is the \emph{state} of the dynamics and $F$ is the \emph{state transition map}. Such a representation is known as \emph{state space} representation in the parlance of dynamical systems theory. When the map $F$ is linear, the dynamics is defined to be linear dynamics. Starting from an initial condition $x(0) = x_0$, its solution under the dynamics (\ref{eq:DS}) is the sequence
\[
x_0, F(x_0), F^{(2)}(x_0), \dots, F^{(k)}(x_0), \dots 
\]
Given the finiteness of the set $\ff^n$, the above sequence can not have infinite terms, and for every initial condition $x_0$, there exists a $L \geq 0$, $M \geq 1$ such that
\[
F^{(k+M)}(x_0) = F^{(k)}(x_0) \quad \quad \forall k \geq L.
\]
$L$ and $M$ are defined as the chain length and the orbit length of the solution starting from $x_0$. It can be seen that when the map $F$ is a permutation of the field $\ff^n$, then $L = 0$ for all initial conditions $x_0$ \cite{Gill2}. The cycle set of the dynamical system (\ref{eq:DS}) is the set of all possible orbit lengths under its dynamics. 

On the study of dynamics over a finite field, the works \cite{Gill1,Harrison} provides a comprehensive literature on the computation of solution structures of a linear dynamical system over a finite field\footnote{Such systems are also known as linear modular systems or  linear sequential machines or linear sequential circuits in the literature.}. For non-singular linear dynamics, given as 
\begin{align}
\label{eq:LDFSS}
x(k+1) = Ax(k),
\end{align}
where $x \in \ff^n$ is the state and $A \in \ff^{n\times n}$ is the state transition map represented as a matrix over $\ff$ and its \emph{cycle set} $\Sigma$ and are computed through elementary divisors of the matrix $A$. 

When the dynamics is non-linear, the computation of the cycle set is a computationally hard problem. Apart from brute force computations, the work \cite{RamSule} gives an algorithmic procedure to estimate the cycle set of a non-linear dynamical system over finite fields by using the Koopman operator and constructing a reduced Koopman operator by restricting the Koopman operator to the smallest invariant subspace containing the coordinate functions and the cycle set of the original non-linear dynamical system is computed through the cycle set of the linear dynamics defined through the reduced Koopman operator, which is well understood (and computable through linear algebraic computations). 

\subsubsection{The Koopman operator framework for dynamical systems}
Given a dynamical system as in (\ref{eq:DS}), define $\ffnz$ to the vector space of all functions from $\ff^n \to \ff$. The Koopman operator $\mckf$ \cite{Koopman} corresponding to (\ref{eq:DS}) is defined over $\ffnz \to \ffnz$ in the following way. Given any $\phi(x) \in \ffnz$,
\[
\mckf \phi(x) = \phi \circ f(x) = \phi(f(x)).
\]
Irrespective of whether the dynamics (\ref{eq:DS}) being linear or not,  the Koopman operator $\mckf$ is a linear operator over the function space $\ffnz$. This linearity of the Koopman operator can be exploited to construct a linear dynamics over $\ffnz$ corresponding to the (\ref{eq:DS}) in the following way
\beq
\label{eq:KLS}
\phi_{k+1}(x) = \mckf \phi_k(x) \quad \quad \phi_0(x) = \phi(x),
\eeq
where each $\phi_k(x) \in \ffnz$. 

Initially, the Koopman operator framework was used extensively for dynamics over reals (or complex) state space, and the function space is infinite-dimensional, which leads to resorting to finite-dimensional numerical approximations of the Koopman operator \cite{Mwilliams, Mauroy} for practical computations. In our setting of dynamical systems over a finite field $\fqn$, the space of functions $\ffnz$ is a finite-dimensional vector space, with dimension $q^n$ and cardinality $q^{q^n}$. The work \cite{RamSule} studies the connection between the dynamics of the non-linear system (\ref{eq:DS}) and the linear system (\ref{eq:KLS}) defined through the Koopman operator. 

\subsection{Representation of a finite group of permutations}
Given a group $G$ of permutations over a finite set, the (group) representation represents the group action in terms of invertible matrices over a finite-dimensional vector space, and the group operation is replaced by matrix multiplication. Such representations are imperative in studying abstract groups as it helps reformulate the problems over the group to equivalent linear algebraic problems. 

A finite field, by definition, is a finite set, and the set of all permutation polynomials over the finite field forms a group under composition. Given a finite subset of such permutations, we can compute a group generated by this set. In this paper, we propose a representation of such a group using the concept of linear representation defined through the Koopman operator. 

\subsection{Outline of the paper}

The paper is organized as follows. Section \ref{sec:LR-ov} focuses on linear representation for maps over finite fields $\ff$, develops conditions for invertibility, computes the compositional inverse of such maps and estimates the cycle structure of permutation polynomials. In Section \ref{sec:LR-ParPP}, this linear representation is extended to a family of parametric maps, studying its invertibility and computation of the parametric inverse. 
The extension of the theory of linear representation to multivariate maps (maps over $\ff^n$) is discussed in Section \ref{sec:LR-mv} and finally, a linear representation of the group generated by a finite set of invertible maps over $\ff^n$ is addressed in Section \ref{sec:LR-FG}.

\section{Linear Representation of a map over the finite field $\fq$}
\label{sec:LR-ov}
Given a map $f:\fq\rightarrow\fq$ and a function $g(x) \in \Funfq$, the iterates of the map $g(x)$ under the Koopman operator $\mckf$ generates the following sequence over $\Funfq$
\begin{equation}
\label{eq:fg-sequence}
g(x), \mckf(g(x)), (\mckf)^2(g(x)), \dots 
\end{equation}
which is the same as 
\[
g(x), g(f(x)),  g(f^{(2)}(x)), \dots .
\]
The initial part of the paper \cite{RamSule} discusses the relationships between the map $f$ and $\mckf$. Since the space $\Funfq$ is of finite dimension (and equal to $q$), there are only finite linearly independent functions in the sequence (\ref{eq:fg-sequence}). This sequence is the \emph{cyclic evolution} of the function $g(x)$ under the linear operator $\mckf$, and the subspace is denoted as $S(\mckf,g)$. When $g(x)$ is the coordinate function (for maps over $\fq$, it is also the identity function) $\chi$ (i.e., $\chi(x) = x$), the cyclic invariant subspace $S(\mckf,\chi)$ is generated by $\chi$ by the action of $(\mckf)^i$, $i=1,2,\dots$ and there exists a smallest $N \in \mathbb{Z}_+$ such that $(\mckf)^{N}(\chi)$ is linearly dependent on the previous functions $(\mckf)^{i}(\chi)$ and $\dim S(f,\chi)=N$. A natural basis for this invariant subspace is 
\begin{align}
\label{eq:basis}
\mcb = \{\chi, \mckf(\chi), (\mckf)^{2}(\chi), \dots, (\mckf)^{N-1}(\chi) \}.
\end{align}
Since $(\mckf)^{N}(\chi)$ is linearly dependent on the previous iterates, we can compute $\alpha_i \in \ff$ such that
\beq\label{alphacoeff}
(\mckf)^{N}(\chi) = \sum_{i = 0}^{N-1} \al_i (\mckf)^{i}(\chi). 
\eeq

Denote by $K_f=\mckf|{S(\mckf,\chi)}$, the restriction of $\mckf$ to this subspace $S$. We call $K_f$, the reduced Koopman operator of $f$ over this invariant subspace $S$. Let $M$ denote the matrix representation of $K_f$ in the basis $\mcb$, be called as the matrix of the reduced Koopman operator. To be consistent with the notation, the matrix $M$ and $K_f$ have the following relationship relative to the basis elements $\mcb$. Denote by $\psi_i=(\mckf)^{(i-1)}(\chi)$ the $i$-th basis function and let $\hat{\psi}$ denote the $N$-tuple column
\[
\hat{\psi}=[\psi_1,\psi_2,\ldots,\psi_N]^T.
\]
Then the matrix representation $M$ of $K_f$ is defined by the relation
\beq\label{MatrixRepNotation}
(K_f)\hat{\psi}=M\hat{\psi}
\eeq
where, the left-hand action of $K_f$ is defined by
\[
(K_f)\hat{\psi}\triangleq [(K_f)\psi_1,\ldots,(K_f)\psi_N]^T.
\]
Hence $M$ is an $N\times N$ companion matrix over $\fq$ in the ordered basis $\psi_i=(\mckf)^{(i-1)}(\chi)$ denoted as
\beq\label{Mmatrix}
M = \begin{bmatrix} 0 & 1 & 0 & \dots & 0 & 0 \\
0 & 0 & 1 & \dots & 0 & 0 \\ 
\vdots & \vdots & \vdots & \ddots & \vdots & \vdots \\
0 & 0 & 0 & \ddots & 0 & 1\\
\al_0 & \al_1 & \al_2 & \dots & \al_{(N-2)} & \al_{(N-1)} \end{bmatrix}
\eeq
where $\alpha_i$ are defined as in equation (\ref{alphacoeff}). This reduced Koopman operator $K_f$ was used to develop a piece of computational machinery to analyze the dynamic evolution of the original map $f(x)$ in the work \cite{RamSule}. In this paper, we show that $M$ inherits important properties of $f$ and define the notion of \emph{linear representation} of the map $f$ using the reduced Koopman operator.  

\subsection{Representation of $f$ in terms of the basis $\mcb$}
We now construct a representation of $f$ itself using the basis $\mcb$ of the invariant subspace $S$. By construction, the space $S$ is the smallest $\mckf$-invariant subspace containing $\chi$ and $K_f$ denotes the action of $\mckf$ on $S$. The coordinate function $\chi$, which belongs to $S$, and has the unique representation in the basis as
\[
\chi=\langle e_1,\hat{\psi}(x)\rangle=e_1^T\hat{\psi}
\]
where, $\langle u,v\rangle=u^Tv$ denotes the Cartesian dot product of $N$-tuples. Then the representation of $f$ itself in the basis $\mcb$ is
\beq\label{reprofonevarf}
f(x)=(\mckf)(\chi)(x)=\langle e_1,M\hat{\psi}\rangle=\langle M^Te_1,\hat{\psi}\rangle=\psi_2.
\eeq

\subsection{Representation of $f^{-1}$}
We now show that whenever $f$ is a permutation function in $\fq$, the inverse function can be represented similarly over the same space $S$. First, we prove a condition of invertibility of $f$ in terms of the representation matrix $M$ in (\ref{reprofonevarf}).

\begin{lemma}
\label{lem:NS_singlevar}
For a map $f: \fq \to \fq$, the following are equivalent
\begin{enumerate}
    \item $f$ is a permutation over $\fq$
    \item $K_f$ is non-singular
    \item $M$ is a non-singular matrix
    \item $\al_0 \neq 0$
\end{enumerate}
\end{lemma}

\begin{proof}
$1 \implies 2$: Let $f$ be a permutation over $\fq$. Then $\mckf$ is non-singular on $\Funfq$ (from Lemma 2 in \cite{RamSule}) hence, the restriction $K_f=\mckf|S$ is also non-singular on $S$. 

$2 \implies 1$: Conversely, let $f$ be not a permutation function. Then there exist $x_1\neq x_2$ in $\fq$, such that $f(x_1)=f(x_2)$. The space $S$ in (\ref{eq:basis}) is generated by $\{\chi,\chi\circ f,\chi\circ f^{(2)},\ldots\}$. Hence the co-ordinate function $\chi$ belongs to $S$. On the other hand
\begin{align*}
f(x_1)=f(x_2) \Rightarrow &\ \mckf \phi(x_1)=\mckf \phi (x_2), \, \forall\, \phi\in\, S\\
  \Rightarrow &\ \psi(x_1) = \psi(x_2) \, \forall\, \psi\in\, K_f(S)\\
  \Rightarrow &\ \chi\, \notin\, K_f(S) \\
   \Rightarrow &\  K_f\, \mbox{ is not one-to-one on }\, S.
\end{align*}
The statements $3,4$ on matrix representation $M$ of $K_f$ and their equivalence with the statement $2$ follows easily. 
 \end{proof}

When $f$ is invertible and $\dim S(\mckf,\chi)=N$, the inverse function $g=f^{-1}$ is represented as follows. Since $M$ is a companion matrix as shown in (\ref{Mmatrix}), its inverse is of the form
\beq\label{Minvmatrix}
M^{-1} = 
\begin{bmatrix} 
c_0 & c_1 & \ldots & c_{(N-2)} & c_{(N-1)} \\
1 & 0 & \dots & 0 & 0\\ 
\vdots & \vdots & \ddots & \vdots & \vdots\\
0 & 0 & \ldots & 1 & 0
\end{bmatrix}. 
\eeq
It can be noted that in the above equation $c_i=-\al_{(i+1)}/\al_0$ for $i=0,\ldots,(N-2)$ and $c_{(N-1)}=1/\al_0$.  
\begin{theorem}
Given a permutation function $f(x)$, and its linear representation (\ref{reprofonevarf}) with the matrix $M$ and its inverse represented in companion form as in (\ref{Mmatrix}) and (\ref{Minvmatrix}) respectively, the compositional inverse $g(x)$ of the function $f(x)$ is given by 
\beq\label{inverforminonevar}
g(x)=\sum_{i=0}^{N-1}c_i(f^{(i)})(x)
\eeq
 Further, the function $g(x)$ has a linear representation over the same basis $\mcb$ given by
\beq\label{repoffinverse}
g(x)=\langle e_1,M^{-1}\hat{\psi}\rangle.
\eeq
\end{theorem}

\begin{proof}
Given the matrices $M$ and $M^{-1}$ as in equations (\ref{Mmatrix}) and (\ref{Minvmatrix}), and from the fact that $MM^{-1} = I$, we have the following relationship from the last row of $MM^{-1}$
\[
\begin{array}{ccl}
\al_0c_0+\al_1 & = & 0\\
\al_0c_1+\al_2 & = & 0\\
 \vdots & \vdots & \vdots\\
\al_0c_{(N-2)}+\al_{(N-1)} & = & 0\\
\al_0c_{(N-1)} & = & 1
\end{array}.
\]
The inverse formula $g(x)$ claimed is equivalent to
\[
g(x)= (\mckf)_g (\chi)(x)=\sum_{i=0}^{N-1}c_i(\mckf)^{i}(\chi)(x).
\]
where $\mckf_g$ is the Koopman operator associated with the function $g(x)$. Hence the formula is verified as
\[
\begin{array}{rcl}
\mckf\ g(x) & = & \sum_{i=0}^{N-1}c_i(\mckf)^{i+1}(\chi)(x)\\
 & = & \sum_{i=0}^{N-2}c_i(\mckf)^{i+1}+c_{N-1}\sum_{i=0}^{N-1}\al_i(\mckf)^i(\chi)(x)\\
 & = & \chi(x)=x
 \end{array}.
\]
The last step follows from the relations between $c_i$ and $\al_i$ above. The representation of $g$ in (\ref{repoffinverse}) is also thus verified.
\end{proof}

\begin{remark}
The above representation of the map $f$ and its inverse is in terms of the linear maps on the smallest $\mckf$-invariant subspace $W$. The formulae for $f$ and $f^{-1}$ are expressed in terms of a chosen basis of $W$. To the best of the authors' knowledge, no other generic formula is available in the literature for expressing $f^{-1}$ directly. It is because of this formula it follows that $f^{-1}$ can be computed in polynomial time if the dimension of $W$ is of polynomial order $O(n^k)$ where $f$ is defined over $\ff_{p^n}$.   
\end{remark}

\subsubsection{Illustrative example}

Consider the map $f(x) = x^3+2x^2+3x+3$ over $\ff_5$. The iterates of $\chi$ are given by functions, 
\[
\chi(x)=x, (\mckf)(\chi)(x)=x^3+2x^2+3x+3,(\mckf)^2(\chi)(x)=2x^3+3x^2+4x+2
\]
while
\[
(\mckf)^3(\chi)(x)=4x+3(\mckf)(\chi)(x)+3(\mckf)^2(\chi)(x).
\]
The basis functions of the cyclic subspace are 
\[
\mcb = \{x, x^3+2x^2+3x+3, 2x^3+3x^2+4x+2 \},
\]
and the matrix representation of the Koopman operator restricted to this subspace with the basis $\mcb$ and its inverse are computed as 
\[
K = \begin{bmatrix} 0 & 1 & 0 \\ 0 & 0 & 1 \\ 4 & 3 & 3  \end{bmatrix} \quad \quad
K^{-1} = \begin{bmatrix} 3 & 3 & 4 \\ 1 & 0 & 0 \\ 0 & 1 & 0 \end{bmatrix}.
\]
The representation of $g(x) := f^{-1}$ in terms of the basis functions $\mcb$ is $[3,3,4]^T$ and
\begin{align*}
g(x) &= 3(x) + 3(x^3+2x^2+3x+3) + 4(2x^3+3x^2+4x+2) \\
&= x^3+3x^2+3x+2
\end{align*}
It can be verified that 
\[
f(g(x)) = g(f(x)) \equiv x
\]
as functions and hence $g(x)$ is the compositional inverse of $f(x)$

\subsection{Linear Representation of Monomials}
In this section, we focus on additional results on the linear representation of $f$ when $f$ is a monomial function. The following theorem re-establishes the condition invertibility of a monomial while adding additional results on the linear complexity. 
\begin{theorem}
Given a monomial $f(x) = x^n$
\begin{enumerate}
\item $f$ is a permutation if and only if $gcd(n,q-1) = 1$.

\item Further when $f$ is a permutation, the linear complexity of $f$ is the multiplicative order of $n$ in $\zq$.
\item When $gcd(n,q-1) \neq 1$ and $n$ belongs to the  the nilradical of $\zqm$, then the linear complexity of $f$ is the index of nilpotence of $n$.
\end{enumerate}
\end{theorem}
\begin{proof}
To prove the results, we first observe that the iterations of $\chi(x) = x$ under $\mckf$ for $f(x) = x^n$ contain only monomials. 
\[
\mckf (\chi) = \mckf (x) = (x)^n
\]
and also
\[
\mckf (x^m) = (x^n)^m = x^{nm\ \modq} \ \ \ \forall\ \ m
\]
Hence the cyclic subspace $S(f,\chi)$ is generated by monomials as follows
\begin{align}
x,x^{n \modq},x^{n^2 \modq},\dots, x^{n^k \modq}.
\label{eq:dummy0}
\end{align} 
So to get a linear dependence in the sequence (\ref{eq:dummy0}), there should be a repetition of a monomial in the sequence. This is because $x^{i}$ and $x^{j}$ are linearly independent as functions whenever $i \modq \neq j \modq$. At some $N$, $x^{n^N \modq} = x^{n^i \modq}$ for some $i < N$ and the $N$ at which a repetition of monomials occurs is the linear complexity of the function $f(x)$.

To prove the first part, let $N$ be the linear complexity of $f$.
\begin{align}
\label{eq:dummy4}
(\mckf)^N \chi = \sum_{i=0}^{N-1} \al_i (\mckf)^i \chi
\end{align}
Since $f$ is a monomial, each of $(\mckf)^k \chi$ is also a monomial. So, only one of the coefficients $\al_i$ in (\ref{eq:dummy4}) is non-zero. However, for invertibility of $f$, from Lemma \ref{lem:NS_singlevar}, it is necessary and sufficient that $\al_0 \neq 0$. Hence 
\begin{align*}
    x^n\ \mbox{is a permutation}\ &\iff \al_0 \neq 0 \\
    & \iff \exists\ N\ \mbox{such that}\ x^{n^N \modq} = x \\
    &\iff n^N \modq = 1 \modq \\
    &\iff gcd(n,q-1) = 1
\end{align*}
which completes the proof of the first statement. 

To prove the second statement, we assume $gcd(n,q-1) = 1$, hence $n$ is a unit of $\zqm$ and let $l$ be the multiplicative order of $n$ in $\zqm^*$.
\begin{align}
\begin{aligned}
& n^l \modq = 1 \\
\implies & x^{n^l \modq} = x \\
\implies & (\mckf)^l \chi = \chi
\end{aligned}
\label{eq:dummy1}
\end{align}
Since there is a linear dependence after $l$ compositions of $f$, the linear complexity of $f$ is $\leq l$. Next, we claim that the linear complexity is exactly $l$. Assuming the contrary, if there exists $0 \leq l_1 < l_2 < l$ such that 
\begin{align*}
&x^{n^{l_1} \modq} = x^{n^{l_2} \modq} \\
\mbox{then},\ & n^{l_1} \modq = n^{l_2} \modq
\end{align*}
Since $n$ and $q-1$ are coprime, it follows that
\begin{align}
&n^{l_2} - n^{l_1} \modq = 0 \nonumber\\
\implies & n^{l_1}(n^{l_2-l_1}-1) \modq = 0 \nonumber \\
\implies & n^{l_2-l_1}-1 \modq = 0  \label{eq:dummy2}
\end{align}
and finally
\[
n^{l_2-l_1} \modq = 1
\]
which is a contradiction to the assumption that $l$ is the multiplicative order of $n$ in $\zqm$ as $l_2 - l_1 < l$ by construction. Hence there cannot be any repetition of any monomial before $x$ in the sequence (\ref{eq:dummy0}), which makes the linear complexity exactly $l$. 



To prove the last part, let $m$ be the index of nilpotence of $n$ in the ring $\zqm$ when $gcd(n,q-1) \neq 1$. 
\begin{align*}
    n^m \modq &= 0  \\ 
    \implies x^{n^m \modq} &= x^0 = \left\{ \begin{matrix} 0 & x = 0 \\ 1 & x \neq 0\end{matrix} \right. \\
\end{align*}
This means that $(\mckf)^{m}(x) = x^{q-1}$ and any further operation of $\mckf$ on $(\mckf)^m$ will remain as $x^{q-1}$ since 
\[
(\mckf)^{m+k}(x)  = ((x^n)^m)^k = (x^{q-1})^k = x^{q-1} \ \ \forall \ \ k \geq 0
\]
Hence any function $x^n$ with $gcd(n,q-1) \neq 1$, under the action of $\mckf$ settles down to the function $x^{q-1}$. Further $m$ is the least such integer such that $n^{m} \modq = 0$ as any smaller $m_1$ such that $x^{n^{m_1}} = x^{q-1}$ is a contradiction to the assumption that $m$ is the index of nilpotence of $n$ in the nilradical of $\zqm$
\end{proof}

\subsection{Cycle structures of a permutation polynomial}
\label{sec:CycStr}
 In this section, we aim to compute the possible cycle lengths of the PP through the linear representation defined in (\ref{reprofonevarf}). As discussed in Section \ref{ss:DSFF}, given a polynomial $f(x)$, we associate a dynamical system through a difference equation of the form 
\beq
\label{eq:DynSys}
x(k+1) = f(x(k)) \ \ \ \ \ x(0) = x_0 \in \fq.
\eeq

Given a permutation polynomial $f(x): \fq \to \fq$, associate a linear dynamics over $\fq^N$ through the reduced Koopman operator in the following way
\beq\label{eq:LDynPP}
y(k+1) = M y(k),
\eeq
where $M$ is defined as in (\ref{Mmatrix}) and $y \in \fq^N$, $N$ being the dimension of the invariant subspace $S(\mckf,\chi)$. Let $\Sigma_M$ be the cycle set of the linear system. The following result characterizes the cycle structures of the permutation polynomial. 
\begin{theorem}
\label{thm:CyclePer}
Given a permutation polynomial $f(x)$ over $\ff$ with the cycle structure $\CSf$  and its linear representation as in (\ref{reprofonevarf}) with matrix $M$. Define linear dynamics over $\ff^N$ as in (\ref{eq:LDynPP}) and let $\Sigma_M$ be the cycle set corresponding to the linear dynamics. Then 
\begin{enumerate}
    \item $\Sigma_f \subset \Sigma_M$.  
    \item Additionally, given $N_i \in \Sigma_M$, it is an element of $\Sigma_f$ if and only if the following system of linear equations 
    \[
    (M^{N_i} - I)y = 0,
    \]
    has a solution $y$ such that 
    \begin{align}
    \label{eq:tempTheorem}
    y = [\psi_1(\alpha),\dots,\psi_N(\alpha)]^T,
    \end{align}
    for some $\alpha \in \ff$. 
    \item For a specific $\alpha \in \ff$, its cycle length under $f$ is equal to the cycle length of the vector $y_\alpha = [\psi_1(\alpha), \dots, \psi_N(\alpha)]^T$ under $M$.
\end{enumerate} 
\end{theorem}

This theorem is a reformulation of the result of Theorem 2 of \cite{RamSule} with the non-linear dynamics be defined over $\ff$ rather than $\ff^n$. We would refer to that manuscript for proof. Further, we point out the following remarks about the theorem. 
\begin{enumerate}
\item The first statement of Theorem \ref{thm:CyclePer} does not imply an equivalence between the cycle structure of the permutation polynomial and the cycle set of the linear dynamics (\ref{eq:LDynPP}), and the former is a subset of the latter. This is because the linear dynamics evolve over a larger set $\ff^N$ than the original PP (defined over $\ff$) and justifies the additional orbits of the linear dynamics. 

\item The second statement of the theorem gives a necessary and sufficient condition for an element of the set $\Sigma_M$ to be in $\Sigma_f$. If the choice of basis is as in (\ref{eq:basis}), once the set of all $y$ satisfying (\ref{eq:tempTheorem}) is obtained, the first component of $y$ is precisely the $\alpha$ for which the other components are to be verified for consistency. 
\end{enumerate}

We look at some additional relations between the sets $\Sigma_f$ and $\Sigma_M$. For a permutation polynomial $f(x)$, define $L_f$ to be the least common multiple (LCM) of all its cycle lengths 
\[
L_f = \mbox{lcm}\ (i\ |\ i \in \Sigma_f).
\]
It can be seen that $L_f$ is the least positive integer such that 
\[ 
f^{(L_f)}(\alpha) = \alpha,\ \ \forall \alpha \in \fq.
\]
Similarly, define $L_M$ to be the LCM of all the cycle lengths of the dynamics defined in (\ref{eq:LDynPP}). From Lemma 4 of \cite{RamSule}, we have the following. 
\begin{theorem}
Given the cycle structure, $\Sigma_f$ of the permutation polynomial and the cycle set $\Sigma_M$ of the associated linear dynamics. Let $L_f$ and $L_M$ be defined as the LCM of all their respective cycle lengths. Then
\[
L_f = L_M.
\]
\end{theorem}

\subsubsection{Numerical Examples}
In this section, we provide examples of estimating the possible orbit lengths of permutation polynomials in the form of Dickson polynomials $D_n(x,\al)$ \cite{MullenPanario} of degree $n$ through the linear representation approach. The Dickson polynomial $D_n(x,\al)$ is of the form 
\[
D_n(x,\al) = \sum_{i = 0}^{\lfloor{\frac{n}{2}}\rfloor} \frac{n}{n-i}\binom{n-i}{i} (-\al)^{i} x^{n-2i}.
\]
\begin{table}[H]
    \centering
\begin{tabular}{|c|c|c|c|c|}
\hline
    \mbox{Permutation polynomial} & \mbox{Finite Field} & \mbox{Dimension of W} & \mbox{True cycle structure} & \mbox{Estimate} \\
    \hline
    $D_{7}(x,4)$  & $\ff_{31}$ & $15$ & $\{ 1,12,16\} $  & $ \{1,4,12,16\}$
    \\
    \hline
    $D_{29}(x,287)$  & $\ff_{307}$ & $153$ & $\{ 1, 53, 200\} $ & $ \{1, 8, 40, 53, 200\}$
    \\
    \hline
    $D_{11}(x,732)$ & $\ff_{1009}$ & $488$  & \shortstack{$\{ 1, 2, 4, 6, 9, 14 $ \\ $ 76, 84, 132, 668\}$} & \shortstack{$\{ 1, 2, 3, 4, 6, 9 $ \\ $ 12, 14, 28, 44, 76, $\\ $ 84, 132, 668\}$} 
    \\
    \hline 
   $D_{5}(x,1)$ & $\ff_{4253}$ & $354$ & $ \{ 1,177,354\}$ & \shortstack{$\{1,2,3,6,59,$\\ $ 118,177,354\}$} \\
    \hline
     \shortstack{$26 x^{27} + 8 x^{22} $\\ $+ 3 x^{12} + 6 x^{7} + 20 x^{2}$} \cite{AMR} & $\ff_{31}$ & $4$ & $\{1,4\}$ & $\{1,2,4\}$ \\
    \hline
\end{tabular}
 \caption{Computation of cycle structure of Permutation Polynomials}
    \label{tab:CycStr}
\end{table}
The work \cite{Mullen2016} also provides a computational framework to compute the cycle structure of the permutation polynomial $f$ by constructing a matrix $A(f)$, of dimension $q\times q$ through the coefficients of the (algebraic) powers of $f^{k}$, $k = 0,1,\dots,q-1$ and computing the multiplicative order of the eigenvalues of this matrix $A(f)$ over a suitable field extension. In our work, to compute the cycle structure of the permutation polynomial, we have to compute the solutions of the associated linear dynamical system (\ref{eq:LDynPP}). This computation amounts to computing the multiplicative order of the eigenvalues of the matrix $M$ over a suitable field extension \cite{Gill1}. From the table, we see that the dimension of the matrix $M$, which is used to compute the cycle lengths, is not necessarily $q$. Hence, this approach does not necessarily involve matrices of dimension $q$ in all cases. 

\section{Linear representation and inversion of parameterized polynomial functions}
\label{sec:LR-ParPP}
We developed a linear representation theory for functions over $\ff$ in the previous section. This section extends the idea to a family of functions over $\ff$ defined through a $\ff$-valued parameter. The well-known Dickson polynomial is one such motivating example for this section. Consider a parameter dependent function $\Fl: \ff \to \ff$, where $\lambda$ is an $\ff$-valued parameter. Any parametric function can also be viewed as a function $F(\lambda,x): \ff^2 \to \ff$. In this section, we develop a linear representation for such functions, explore the parametric dependent invertibility of $\Fl$, and construct the parametric inverse of $\Fl$. 

A map $\Fl$ with an $\ff$-valued parameter is defined to be \emph{parametrically invertible} over $\ff$ if the map $\fb := F_{\lambda = \beta}$ is invertible for each $\beta \in \ff$. We define the \emph{parametric inverse} of $\Fl$ as another paramteric function $G_\lambda(x)$ (with the same parameter $\lambda$), such that $G_\lambda(x)$ satisfies
\[
\Fl \circ G_\lambda := F(\lambda,G(\lambda,x)) = Id_{\ff}(x) = G(\lambda,F(\lambda,x)) =: G_\lambda \circ \Fl 
\]
where $Id_{\ff}(x)$ is the identity map over $\ff$ (and is independent of the parameter $\lambda$).

Any parameterized function $\phi_{\lambda} : \fq \to \fq$ can be written in the form 
\begin{align}
\begin{aligned}
\phi_{\lambda}(x) &= \sum_{i = 0}^{q-1} p_i(\lambda) x^i \\ 
&=\begin{bmatrix} p_{0}(\lambda) & p_1(\lambda) & \dots & p_{q-1}(\lambda) \end{bmatrix} \begin{bmatrix} 1 \\ x \\ \vdots \\ x^{q-1}\end{bmatrix}  \\
&= v_{\phi}(\lambda)^T X
\end{aligned}
\label{eq:PF_mat_eq}
\end{align}
where $v_\phi (\lambda)= [ p_{0}(\lambda) \ p_1(\lambda) \ \dots \ p_{q-1}(\lambda) ]^T $ and $X = [1 \ x \ x^2 \ \dots \ x^{q-1}]^T$. This gives a $1-1$ correspondence between any parameteric function $\phi_{\lambda}(x)$ with $v_{\phi}(\lambda)$. Each of $p_i(\lambda)$ in (\ref{eq:PF_mat_eq}) can be considered as a rational function over $\ff$ and hence $p_i(\lambda) \in \mcR$, the ring of fractions over $\ff[\lambda]$ defined as 
\[
\mcR = \{ \frac{f}{g} \ \ | f,g \in \ff[\lambda]\ \ \mbox{and}\ \ g \in S\}
\]
where $S$ is the set of \emph{non-vanishing functions} over $\ff$ defined as 
\[
S = \{ g \in \ff[\lambda] | g(\alpha) \neq 0 \ \ \forall\ \alpha \in \ff \}
\]
and this set $S$ is a closed under multiplication. It follows clearly that $\mcR$ is a principal ideal domain (PID) and the set of rational functions $h = f/g$ in which both $f$ and $g$ belong to $S$ form the units of $\mcR$. Each $v_\phi(\lambda)$ given in (\ref{eq:PF_mat_eq}) is an element of $\mcR^{q \times 1}$, which is a finitely generated, and free module $\mM$ over $\mcR$. When $\phi_{\lambda}(x) = \Fl(x)$, we get $v_F$, the parameteric coeffiecients of the function $\Fl$. 

\subsubsection*{Koopman operator for the map $\Fl$}
Next we define the Koopman operator $\mckfl$ for the parametric map $\Fl$ and its action on functions $\phi_\lambda:\ff \times \ff \rightarrow \ff$. Given any function $\phi_{\lambda} = \phi(\lambda,x)$, we define the action of $\mckfl$ acting on $\phi_\lambda$ in the following way 
\begin{align}
\label{eq:ParaDual}
\mckfl\phi_\lambda = \phi(\lambda,\Fl(x))
\end{align}
Equivalently, $\mckfl$ can be thought of as an operator acting on $v_\phi$ in the following way
\[
\mckfl (v_\phi^T X) = v_\phi^T \mckfl X =: v_{\Fl \phi}^T X.
\]
where $v_{\Fl \phi}^T$ is the parametric coefficients of $\phi(\lambda,\Fl(x))$ defined in (\ref{eq:ParaDual}). 
\subsubsection*{$\mcR$-Linearity of $\mckfl$}
Given any two $\phi_1(\lambda,x)$ and $\phi_2(\lambda,x)$,
\begin{align*}
\mckfl(\al(\lambda)\phi_1(\lambda,x)+\beta(\lambda)\phi_2(\lambda,x)) &=  \alpha(\lambda)\phi_1(\lambda,\Fl(x)) +  \beta(\lambda)\phi_2(\lambda,\Fl(x)) \\
&= \alpha(\lambda)\mckfl \phi_1(\lambda,x) +\beta(\lambda) \mckfl \phi_2(\lambda,x)
\end{align*}

This makes $\mckfl$ a $\mcR$-linear map over parametric maps $\phi_\lambda$, or equivalently a $\mcR$-linear map over $v_\phi$ in $\mM$. Hence there exists a matrix $M_\lambda$ over $\mcR$ such that the action of $\mckfl$ in (\ref{eq:ParaDual}) can be written as
\begin{align}
\label{eq:Fl_mat}
   \mckfl v_\phi = M_\lambda v_\phi
\end{align}


\subsection{Linear representation of parameterized functions}
Next, a notion of linear representation is developed for parametrized functions. Starting from the coordinate function $\chi(\lambda,x) := x$, under the action of $\mckfl$ as defined in (\ref{eq:ParaDual}), the following sequence can be generated
 \begin{align}
 \label{eq:Fl_seq}
 \chi(x) = x \xrightarrow{\mckfl} \Fl(x)  \xrightarrow{\mckfl} \Fl(\Fl(x)) = \Fl^{(2)} (x) \xrightarrow{\mckfl} \cdots 
 \end{align}
where each of the functions $\Fl^{(i)}(\chi)$ can be viewed as parameterized functions. Corresponding to the sequence of parametrized functions
\[
\chi, \mckfl (\chi), (\mckfl)^{2} (\chi), \dots 
\]
one can associate a sequence 
\begin{align}
v_{F_0}, v_{F_1}, v_{F_2}, \dots 
\label{eq:VF_seq}
\end{align}
where $v_{F_i} \in \mM$ are computed as in (\ref{eq:PF_mat_eq}) with $\phi_{\lambda}(x) = \Fl^{(i)}(x)$.

\begin{proposition}
There exists a smallest $N$ such that the $v_{F_N}$ as defined in the sequence (\ref{eq:VF_seq}) satisfies  
\begin{align}
 v_{F_N} = \sum_{i=0}^{N-1} \alpha_i(\lambda) v_{F_i}
\label{eq:LD_coeff_Fl}
\end{align}
$\al_i(\lambda) \in \ff(\lambda)$, where $\ff(\lambda)$ is the field of rational fractions and $v_{F_i}$, $i= 0,\dots, N-1$ are $\ff(\lambda)$-linearly independent.
\end{proposition}
\begin{proof}
Given that $\mcR$ is a PID, we can embed $\mcR$ into the field of rational fractions{\footnote{The field of rational fractions $\ff(\lambda) = \{ f/g, f,g \in \ff[\lambda],  g\neq 0$\}}} $\ff(\lambda)$. The linearity of the operator $K_\lambda$ carries over to the field $\ff(\lambda)$. Considering $\tilde{\mM}$ as an extension of the $\mcR$-module $\mM$ as a vector space over $\ff(\lambda)$, we know that $v_{F_0}$ has a minimal annihaliting polynomial (MAP) $m(X)$ over the field $\ff(\lambda)$ which is a monic polynomial of degree $N$ in the indeterminate $X$ such that 
\[
m(K_\lambda) v_{F_0} = 0 
\]
Considering $m(X)$ to be 
\begin{align}
\label{eq:MAP}
m(X) = X^N - \sum_{i = 0}^{N-1} \alpha_{N-i} X^{N-i}\ \  \mbox{with} \ \ \alpha_i \in \ff(\lambda)
\end{align}
and since $v_{F_0} \in \tilde{\mM}$, we have  
\[
m(K_\lambda) v_{F_0} = 0 \implies \bigg( K_\lambda^N - \sum_{i=0}^{N-1} \alpha_{N-1} (K_\lambda)^{N-i} \bigg)v_{F_0} = 0
\]
and from (\ref{eq:Fl_mat}) we see that 
\[
(K_\lambda)^i v_{F_0} = v_{F_i}
\]
Hence we have the relation 
\[
v_{F_N} = \sum_{i=0}^{N-1} \alpha_i(\lambda) v_{F_i}
\]
\end{proof}

\begin{lemma}
Given the sequence of iterates $\Fl^{(i)}(x)$ of $\Fl(x)$ under $\mckfl$ as in (\ref{eq:Fl_seq}), and its corresponding $v_{F_i}$ as in (\ref{eq:VF_seq}), if $v_{F_N}$ satisfies (\ref{eq:LD_coeff_Fl}) then 
\begin{align}
\label{eq:fl_rec} 
\Fl^{(N)}(x) = \sum_{i = 0}^{N-1} \alpha_i(\lambda) \Fl^{(i)}(x)
\end{align}
\end{lemma}
\begin{proof}
Given that $\Fl^{(N)}$ is represented as $v_{F_N} \in \mcR^{1 \times q}$, we have
\begin{align*}
    \Fl^{(N)}(x) &= v_{F_N}^T \begin{bmatrix} 1 \\ x \\ \vdots \\ x^{q-1}\end{bmatrix} \\ 
    &= \big( \sum_{i=0}^{N-1} \alpha_i(\lambda) v_{F_i}^T \big) \begin{bmatrix} 1 \\ x \\ \vdots \\ x^{q-1}\end{bmatrix} \\ 
    &= \sum_{i=0}^{N-1} \big( \alpha_i(\lambda) v_{F_i}^T \begin{bmatrix} 1 \\ x \\ \vdots \\ x^{q-1}\end{bmatrix} \big)\\ 
    &= \sum_{i=0}^{N-1} \alpha_i(\lambda) \Fl^{(i)}(x)
\end{align*}
\end{proof}
\subsection*{Linear representation and complexity of $\Fl(x)$}
Given  $\hsl := [\chi, \mckfl \chi, \dots, (\mckfl)^{(N-1)} \chi]^T$, the action of $\mckfl$ on $\hsl$ can be written as 
\begin{align}
\label{eq:Mlambda1}
\mckfl \hsl = \Ml \hsl
\end{align}
where $\Ml \in \ff(\lambda)^{N \times N}$ and defined as
\beq
\label{eq:Mlambda}
M_\lambda = \begin{bmatrix} 0 & 1 & 0 & \dots & 0 & 0 \\
0 & 0 & 1 & \dots & 0 & 0 \\ 
\vdots & \vdots & \vdots & \ddots & \vdots & \vdots \\
0 & 0 & 0 & \ddots & 0 & 1\\
\al_0(\lambda) & \al_1(\lambda) & \al_2(\lambda) & \dots & \al_{N-2}(\lambda) & \al_{N-1}(\lambda) \end{bmatrix}
\eeq
where $\alpha_i(\lambda)$ is defined as in equation (\ref{eq:LD_coeff_Fl}). This helps in defining the \emph{linear representation} of $\Fl(x)$ analogous to the one defined for $f(x)$ in (\ref{reprofonevarf}) in the following way
\beq
\Fl(x) = \mckfl \chi = \langle e_1, \Ml \hsl\rangle = \langle \Ml^T e_1, \hsl \rangle = \Fl(x)
\eeq
The degree $N$ of the MAP $m(X)$ is defined to be the \emph{linear complexity} of the parameterized function $\Fl(x)$. 



\subsection{Parametric invertibility}
The following theorem establishes a condition for a given function $\Fl$ to be parametrically invertible.
\begin{theorem}
Given a parameterized function $\Fl$, $\Fl$ is parametrically invertible over $\fq$ if and only if $\alpha_0$ is an unit of $\mcR$.
\end{theorem}

\begin{proof}

Sufficiency: Assuming $\al_0(\lambda)$ is an unit of $\mcR$, the inverse of $\Ml$ exists over $\ff(\lambda)$. Construct 
\begin{align}
\label{eq:Fl-inverse}
\Gl =  \langle e_1, \Ml^{-1} \hsl \rangle = e_1^T \Ml^{-1} \hsl 
\end{align}
It can be verified that 
\begin{align*}
\Fl(\Gl) = \mckfl (\Gl) &= \mckfl (e_1^T \Ml^{-1} \hsl)  \\
    & = e_1^T \Ml^{-1} \mckfl \hsl \quad \mbox{By linearity of}\ \mckfl\\
    &= e_1^T \Ml^{-1} \Ml \hsl \quad \mbox{By equation (\ref{eq:Mlambda1})} \\
    & = e_1^T \hsl \\
    &= \chi(x) = Id_{\ff}
\end{align*}
Necessity: Choosing a $\be \in \ff$, let $\fb$ be the function corresponding to $\Fl$ when $\lambda = \be$. Such an $\fb$ is a function from $\ff \to \ff $, and from the previous section, the linear representation of $\fb$ can be computed using $\mckf_\beta$, the Koopman operator corresponding to the map $\fb$. Similar to the equation (\ref{alphacoeff}), one can write   
\[
(\mckf_\beta)^{n_\beta} = \sum_{i=0}^{n_\beta-1} c_{\beta,i} (\mckf_\be)^i
\]
where $n_\beta$ is the linear complexity of $\fb$ and $c_{\beta,i} \in \ff$. This relation is expressed alternatively as polynomial $p_\beta(x)$
\begin{align}
\label{eq:pi_poly}
p_\be(X) = X^{n_\be} - \sum_{i=0}^{n_\be-1} c_{\be,i} X^i
\end{align}
as the minimal annihilating polynomial of $\mckf_\beta$. 

Consider the polynomial $q(X)$ 
\[
q(X) = \prod_{\be \in \ff} p_\be(X).
\]
Given any polynomial $p(X) = \sum_{i=0}^n c_i X^i$, with $c_i \in \ff(\lambda)$ define 
\[
p(\mckfl)  = \sum_{i=0}^n c_i (\mckfl)^i
\]
We claim that $q(\mckfl) = 0$: This is because for each $\be \in \ff$
\begin{align*}
q(\mckf_\be) &=  \bigg( \prod_{\hat{\beta} \in \{\ff - \be\}} p_{\hat{\beta}}(\mckf_\be) \bigg) p_\be(\mckf_\be) \\  
    &=  \bigg( \prod_{\hat{\be} \in \{\ff - \be\}} p_{\hat{\be}}(\mckf_{\be}) \bigg) (0) \\
    &= 0
\end{align*}
and $q(X)$ is an annihilator for each of $\mckf_\be$ and hence it is an annihilator for $\mckfl$. 
But from (\ref{eq:fl_rec}), we know that $m(X)$ 
\[
m(X) = X^N - \sum_{i=0}^{N-1} \al_i(\lambda) X^i, 
\]
is the smallest degree annihilating polynomial for $\mckfl$. Hence $\mbox{deg}(m) \leq \mbox{deg}(q)$ and $m(X)|q(X)$. Also 
\[
m_{\beta}(X) := X^N-\sum_{i=0}^{N-1} \al_i(\lambda = \beta) X^i
\]
satisfies
\[
m_{\be}(\mckf_\beta) = 0
\]
and
\[
m_{\be} (X) | q(X)
\]
Since $\Fl$ is invertible, $\fb$ is also invertible. Hence, each $p_\be(X)$ in equation (\ref{eq:pi_poly}) has a non-zero constant term (from Lemma \ref{lem:NS_singlevar}), and hence $q(X)$ also has a non-zero constant term (as it is a product of polynomials, each having a non-zero constant term). This makes the constant term of $m_{\be} (X) = \al_0(\lambda = \be)$ to be non-zero. 

Since this is true for all $\be \in \ff$, the $\al_0(\lambda)$ is non-vanishing all over $\ff$ and hence it is an unit of $\mcR$. 
%
\end{proof}

In many cases, the map $\Fl$ may not be parametrically invertible, but it is still important to know if it is invertible for some specific values of $\lambda \in \ff$. The following corollary follows from the above theorem and establishes the invertibility of the map $\Fl$ for specific $\lambda$. 

\begin{corollary}
Given a parameterized function $\Fl$ and its linear representation as $\Ml$, $\Fl$ is invertible for a specific $\lambda = a \in \fq$ if and only if $M_{\lambda = a}$ is a non-singular matrix or equivalently $\al_0(a) \neq 0$.
\end{corollary}

\begin{remark}
Given a parameterized PP $\Fl$ over $\fq$, for each $\lambda_0 \in \fq$, its cycle structure can be computed to be $\Sigma_{\lambda_0}$. An estimate of $\Sigma_{\lambda_0}$ can be computed through the corresponding cycle set of the matrix $M_{\lambda = \lambda_0}$ as discussed in Section \ref{sec:CycStr}.  
\end{remark}

\subsection{Numerical examples}
The linear representation of maps and their inverses to two specific cases of parameterized permutation polynomials is constructed as an application. The first one is a parameterized polynomial over $\ff_{13}$ from the paper \cite{ZhengWangWei}, and the second one is a Dickson Polynomial over $\ff_{17}$. All the computations are carried out using the open-source software SAGEmath. 

\subsubsection{Parameterized polynomial over $\ff_{13}$} 
In the paper by Zheng Et al. \cite{ZhengWangWei}, it is shown that the (normalized) polynomial (i.e., the polynomial function) $f_a(x) = x^5+ax^3+3a^2x$ is a permutation over $\ff_{13}$ whenever $a$ is a non-square element of $\ff_{13}$. The non-square elements of $\ff_{13}$ are $2, 5, 6, 7, 8, 11$. 

The linear representation for $f_a(x)$ is constructed using the abovementioned theory. Starting from $\psi_1(x) = \chi(x) = x$, the sequence $\psi_i(x) = \mckf_a \psi_{i-1}(x)$ is computed.
\begin{align}
\label{eq:psi-def}
\begin{aligned}
    \psi_1(x) &= x , \quad \quad \quad \quad \psi_2(x) = x^{5} + ax^{3} + 3a^2 x \\
     \psi_3(x) &= 6a^{7} x^{11}  -  a^{8}x^{9}  + 2a^{9} x^{7}  + 9a^{10} x^{5}  -  a^{3}x^{11}  + 6a^{4} x^{9}  + 5a x^{11}  + 10 a^{5} x^{7} - a^{2}x^{9}  +a^{6} x^{5}  \\ 
    &\quad  + 5a^{3} x^{7}  +a^{7} x^{3}  + 3 a^{4}x^{5}  + 6  a^{5}x^{3} + 3 a^{2}x^{5}  + 9 a^{6}x  + 3 a^{3}x^{3}  + 9a^{4} x  + ax^{3}  + 3a^{2} x  + x \\
    \psi_4(x) &= 6a^{11} x^{11}  + 10a^{12} x^{9}  + 8 a^{10}x^{9}  + 6  a^{7}x^{11} + 10a^{11} x^{7}  + 6a^{8} x^{9}  + 7a^{12} x^{5}  + 2a^{5} x^{11}  + 8a^{9} x^{7}   \\
    &\quad + 6a^{6} x^{9}  + 11a^{3} x^{11}  + a^{7}x^{7}  + 4 a^{11}x^{3}  + a^{4}x^{9}  + 9  a^{8}x^{5} + 8 a x^{11} + 10 a^{5}x^{7}  + 9  a^{9}x^{3} - a^{2} x^{9}   \\ 
    &\quad +  a^{6}x^{5} + 5 a^{10} x + 8 a^{3} x^{7} + 5 a^{7}  x^{3}+ 2 a^{4}x^{5}  + 8a^{8} x  + 10 a x^{7} + 7a^{5} x^{3}  -  a^{2}x^{5}  + 6a^{6} x  \\
    &\quad + 9 a^{3} x^{3} + x^{5} + 6a^{4} x  + 8  ax^{3} + 6 a^{2} x  \\
     \psi_5(x) &= 9 a^{11}x^{11}  + 2a^{12} x^{9}  + 10 a^{9} x^{11} + 10a^{10} x^{9}  +  a^{11}x^{7} + 4 a^{8} x^{9} + 11a^{12} x^{5}  + 11 a^{9} x^{7} + 8a^{6} x^{9}  \\ 
    &\quad + 5a^{10} x^{5}  -  a^{3} x^{11} + 10 a^{7} x^{7} + 6  a^{11}x^{3} + a^{4} x^{9} + 2a^{8} x^{5}  + 11 a x^{11} + 10a^{9} x^{3}  + 10 a^{2}x^{9}   \\ 
    &\quad + 9a^{6} x^{5}  + 9a^{10} x  + 9a^{3} x^{7}  + 3a^{7} x^{3}  + 3  a^{4}x^{5} + 11a^{8} x  + 4 a x^{7} + 7a^{5} x^{3}  + 6 a^{2} x^{5} + 10a^{6} x  \\
    &\quad + 7 a^{3}x^{3}  + 2a^{4} x  + 8a x^{3}  + 10a^{2} x  + x  \\
    \psi_6(x) &= 6a^{11} x^{11}  + 11  a^{12}x^{9} +  a^{9}x^{11} + 7 a^{10}x^{9}  + 9a^{7} x^{11}  + 8a^{11} x^{7}  + 11a^{8} x^{9}  + 7a^{12} x^{5}  + 6a^{5} x^{11}   \\ 
    &\quad + a^{9}x^{7}  + 11a^{6} x^{9} + 11 a^{10} x^{5} +  a^{3}x^{11} + 6 a^{7}x^{7}  + 10 a^{11}x^{3}  + 8a^{4} x^{9}  + 4a^{12} x  + 9a x^{11}  \\
    &\quad + 11a^{5} x^{7}  + a^{9}x^{3}  + 10a^{6} x^{5}   + a^{3}x^{7}  + 8 a^{7}x^{3}  + 10a^{8} x  + 4  ax^{7} -   a^{5}x^{3} + 4a^{2} x^{5}  + 3 a^{6}x  \\
    &\quad + 2a^{3} x^{3}  + x^{5} -  a^{4}x  + 3a x^{3}  + 5 a^{2} x
\end{aligned}
\end{align}
It is computed that $\psi_7(x) = \mckf_a \psi_6(x)$ can be written as 
\[
\psi_7(x) = \sum_{i = 1}^{6} \al_i \psi_i(x)
\]
This is a matrix $M$ over $\ff_{13}^{6 \times 6}$
\[
M = \left(\begin{matrix}
0 &1 & 0 & 0 & 0 & 0  \\
0&  0 & 1 & 0 & 0 & 0   \\
0 & 0 & 0& 1 & 0 & 0   \\
0 & 0 &0 & 0 & 1 & 0 \\
0 & 0 & 0 & 0&0&1  \\
\al_0 &\al_1 & \al_2 & \al_3 & \al_4 & \al_5 
\end{matrix}\right)
\]
where 
\begin{align*}
    \al_0 &= \frac{5 a^{28} + 7 a^{26} + 2 a^{22} + 2 a^{20} + 5 a^{18} -  a^{16} + 5 a^{14} + 5 a^{8} + 10 a^{6} + 7 a^{4} + 7 a^{2} + 11}{6 a^{16} + 3 a^{14} + 9 a^{12} + 5 a^{10} + 11 a^{8} + 5 a^{6} -  a^{4} + 6 a^{2}} \\
    \al_1 &= \frac{9 a^{26} + a^{24} + 11 a^{22} + 9 a^{20} + 3 a^{16} + 9 a^{14} + 10 a^{12} + 2 a^{10} -  a^{8} + 9 a^{6} + 10 a^{4} + 6}{6 a^{16} + 3 a^{14} + 9 a^{12} + 5 a^{10} + 11 a^{8} + 5 a^{6} -  a^{4} + 6 a^{2}} \\
    \al_2 &= \frac{- a^{16} + 5 a^{14} + 4 a^{12} + 9 a^{8} -  a^{6} + 9 a^{4} + 2 a^{2} + 7}{8 a^{14} + 4 a^{12} -  a^{10} + 11 a^{8} + 6 a^{6} + 11 a^{4} + 3 a^{2} + 8} \\
    \al_3 &=  \frac{8 a^{16} + 5 a^{14} -  a^{12} + 9 a^{8} + 2 a^{6} + 5 a^{4} + 7 a^{2} + 7}{6 a^{16} + 3 a^{14} + 9 a^{12} + 5 a^{10} + 11 a^{8} + 5 a^{6} -  a^{4} + 6 a^{2}} \\
    \al_4 &= \frac{- a^{16} + 4 a^{14} + 11 a^{12} + 5 a^{10} + 10 a^{8} + 7 a^{6} + 7 a^{4} + 8 a^{2} + 2}{6 a^{16} + 3 a^{14} + 9 a^{12} + 5 a^{10} + 11 a^{8} + 5 a^{6} -  a^{4} + 6 a^{2}} \\
    \al_5 &=  \frac{3 a^{14} + 2 a^{12} + 5 a^{10} + 7 a^{8} + 8 a^{4} + 8 a^{2} + 10}{8 a^{14} + 4 a^{12} -  a^{10} + 11 a^{8} + 6 a^{6} + 11 a^{4} + 3 a^{2} + 8}.
\end{align*}
Assuming $\hat{\psi} = [\psi_1(x), \psi_2(x), \psi_3(x), \psi_4(x), \psi_5(x), \psi_6(x)]^T$, where $\psi_i(x)$ are defined in (\ref{eq:psi-def}), the linear representation of the parametric function $f_a(x)$ is
\begin{align}
\begin{aligned}
f &= e_1^T M \hat{\psi} \\
 &= \begin{bmatrix} 1 & 0 & 0 & 0 & 0&0 \end{bmatrix} M \hat{\psi} \\
 &= \begin{bmatrix} 0 & 1 & 0 & 0 & 0 & 0 \end{bmatrix} \hat{\psi}
 \end{aligned}
\end{align}
It can be seen that $\alpha_0$ is not an unit of $\mcR$ as the denominator has $a^2$ as a factor. The determinant of this matrix is
\[
\mbox{det} \ M = \frac{8 a^{28} + 6 a^{26} + 11 a^{22} + 11 a^{20} + 8 a^{18} + a^{16} + 8 a^{14} + 8 a^{8} + 3 a^{6} + 6 a^{4} + 6 a^{2} + 2}{6 a^{16} + 3 a^{14} + 9 a^{12} + 5 a^{10} + 11 a^{8} + 5 a^{6} -  a^{4} + 6 a^{2}}
\]
Since $a \in \ff_{13}$, using the fact that $a^{13} = a$, the determinant of $M$ can be written in terms of irreducible factors as 
\begin{align*}
\mbox{det}\ M = \frac{(a + 4)(a + 9)(a + 10)(a + 12)(a + 1)(a + 3)(a^{2} -  a + 12)(a^{2} + a + 12)}{2 a^{2}  (a^{2} -  a + 3) (a^{2} + a + 3) (a^{3} -  a^{2} -  a + 6) (a^{3} + a^{2} -  a + 7)}
\end{align*}

It can be clearly seen that the matrix $M$ has zero determinant for $a = 1,3,4,9,10,12$ in $\ff_{13}$ and for $a = 0$, the operator M is ill-defined as the determinant does not exist, and hence the determinant of $M$ is non-zero only when $a = 2, 5, 6, 7, 8, 11$ which are precisely the non-square elements of $\ff_{13}$. 

\begin{remark}
Though the formula given in equation (\ref{eq:Fl-inverse}) is valid only when $\alpha_0$ is an unit of $\mcR$, one can, in principle, construct the formal inverse of $\Ml$ and use the fact that such a formula is valid only for those parameters $\beta$ where $M_{\lambda = \beta}$ is an invertible matrix over $\ff$.  
\end{remark}

\begin{align}
\begin{aligned}
    f_a^{-1}(x) &= e_1 ^T M^{-1} \hat{\psi} \\
    &= \begin{bmatrix} c_0 & c_1 & c_2 & c_3 & c_4 & c_5\end{bmatrix} \hat{\psi}
\end{aligned}
\end{align}
where 
\begin{align*}
    c_0 &= \frac{4 a^{20} -  a^{18} + 2 a^{16} + 8 a^{14} -  a^{12} -  a^{10} -  a^{8} + 2 a^{6} + 10 a^{4} + 6}{5 a^{22} + 7 a^{20} + 7 a^{16} + 9 a^{14} + 5 a^{12} + 6 a^{10} + a^{8} + 5 a^{6} + 6 a^{4} + 6 a^{2} + 2} \\
    c_1 &= \frac{9 a^{16} + 8 a^{14} + a^{12} + 3 a^{10} + 6 a^{8} + a^{6} + 8 a^{2}}{8 a^{26} + 4 a^{24} -  a^{22} + 8 a^{20} + 9 a^{18} + 9 a^{16} + 2 a^{14} + a^{12} + 3 a^{10} + 9 a^{8} + 9 a^{6} + 4 a^{4} + 5 a^{2} + 8} \\
    c_2 &= \frac{5 a^{16} + 8 a^{14} + a^{12} + 4 a^{8} + 11 a^{6} + 8 a^{4} + 6 a^{2} + 6}{5 a^{28} + 7 a^{26} + 2 a^{22} + 2 a^{20} + 5 a^{18} -  a^{16} + 5 a^{14} + 5 a^{8} + 10 a^{6} + 7 a^{4} + 7 a^{2} + 11} \\
    c_3 &= \frac{a^{16} + 9 a^{14} + 2 a^{12} + 8 a^{10} + 3 a^{8} + 6 a^{6} + 6 a^{4} + 5 a^{2} + 11}{5 a^{28} + 7 a^{26} + 2 a^{22} + 2 a^{20} + 5 a^{18} -  a^{16} + 5 a^{14} + 5 a^{8} + 10 a^{6} + 7 a^{4} + 7 a^{2} + 11} \\
    c_4 &= \frac{- a^{16} + 8 a^{14} + 7 a^{12} + 2 a^{10} + 6 a^{6} + 6 a^{4} + a^{2}}{8 a^{28} + 6 a^{26} + 11 a^{22} + 11 a^{20} + 8 a^{18} + a^{16} + 8 a^{14} + 8 a^{8} + 3 a^{6} + 6 a^{4} + 6 a^{2} + 2} \\
    c_5 &= \frac{6 a^{16} + 3 a^{14} + 9 a^{12} + 5 a^{10} + 11 a^{8} + 5 a^{6} -  a^{4} + 6 a^{2}}{5 a^{28} + 7 a^{26} + 2 a^{22} + 2 a^{20} + 5 a^{18} -  a^{16} + 5 a^{14} + 5 a^{8} + 10 a^{6} + 7 a^{4} + 7 a^{2} + 11}
    \end{align*}
Simplifying the above expression (by reducing the degree of $a$ by using $a^{13} = a$), the function $f_a^{-1}(x)$ is given below.
\begin{align}
\label{eq:ExZWZ1}
\begin{aligned}
 f_a^{-1}(x) = &\frac{1}{d} \Bigg[ (10 a^{11} + 7 a^{9} + a^{7} + 10 a^{5} + 7 a^{3} + a) x^{11} + (10 a^{10} + 6 a^{8} + 4 a^{6} + 11 a^{4} + 10 a^{2}) x^9 + \\
 &(6 a^{11} + 9 a^{9} + 2 a^{7} + 10 a^{5} + 10 a^{3} + 6 a) x^7 + (a^{12} + 9 a^{10} - a^{8} + 2 a^{6} + 6 a^{4} + 8 a^{2} + 11) x^5 \\
 & + (3 a^{11} + 7 a^{9} + 2 a^{7} + 10 a^{3} + 6 a) x^3 + (7 a^{12} + 7 a^{10} + 5 a^{8} -  a^{6} + a^{4} + 11 a^{2}) x \Bigg]    
\end{aligned}
\end{align}
where,
\begin{align}
d = 2 (a + 4)(a + 9)(a + 10)(a + 12)(a + 1)(a + 3)(a^{2} -  a + 12)(a^{2} + a + 12)
\label{eq:ExZWZd}
\end{align}
As discussed earlier, this function is not defined for $a \in \ff_{13}$, which are square elements. It can be verified that, indeed, these square elements are exactly the roots of linear factors of $d$, and for the non-square elements $a \in \ff_{13}$, the inverse of $f$ is given by the formula (\ref{eq:ExZWZ1}). To complete the example, the paper \cite{ZhengWangWei} gives the formula for inverse as 
\begin{align}
\label{eq:exZWZ2}
f_a^{-1}(x) = -a^2 x^9 - ax^7 + 4x^5 +4 a^5x^3 -5a^4x
\end{align}
Though these two do not appear to be the same polynomials, it can be seen in the following table that these two equations (\ref{eq:ExZWZ1}) and (\ref{eq:exZWZ2}) lead to the same functions over $\ff_{13}$ whenever $a$ is a non-square element of $\ff_{13}$.

\begin{table}[H]
    \centering
    \begin{tabular}{|c|c|c|}
    \hline
         $a$ & $f_a^{-1}(x)$ from equation (\ref{eq:ExZWZ1})  & $f_a^{-1}(x)$ from equation (\ref{eq:exZWZ2})   \\
         \hline 
         2 &  $ 9 x^{9} + 11 x^{7} + 4 x^{5} + 11 x^{3} + 11 x$ & $ 9 x^{9} + 11 x^{7} + 4 x^{5} + 11 x^{3} + 11 x$ \\
         \hline 
         5 & $x^{9} + 8 x^{7} + 4 x^{5} + 7 x^{3} + 8 x$& $x^{9} + 8 x^{7} + 4 x^{5} + 7 x^{3} + 8 x$ \\
         \hline 
         6 & $3 x^{9} + 7 x^{7} + 4 x^{5} + 8 x^{3} + 7 x$ & $3 x^{9} + 7 x^{7} + 4 x^{5} + 8 x^{3} + 7 x$ \\
         \hline 
         7 & $3 x^{9} + 6 x^{7} + 4 x^{5} + 5 x^{3} + 7 x$ & $3 x^{9} + 6 x^{7} + 4 x^{5} + 5 x^{3} + 7 x$ \\
         \hline 
         8 & $x^{9} + 5 x^{7} + 4 x^{5} + 6 x^{3} + 8 x$ & $x^{9} + 5 x^{7} + 4 x^{5} + 6 x^{3} + 8 x$ \\
         \hline 
         11 & $ 9 x^{9} + 2 x^{7} + 4 x^{5} + 2 x^{3} + 11 x$ & $ 9 x^{9} + 2 x^{7} + 4 x^{5} + 2 x^{3} + 11 x$ \\
         \hline
    \end{tabular}
    \caption{The function $f_a^{-1}(x)$ for non-square elements $a$ of $\ff_{13}$}
    \label{tab:verification}
\end{table}


\subsubsection{Example 2 : Dickson Polynomial over $\ff_{17}$}
Consider the Dickson Polynomial of degree $11$ over $\ff_{17}$
\begin{align}
\label{eq:DSon1}
D_{11}(x,a) = x^{11} + 6a x^{9} + 10a^{2}  x^{7} + 8a^{3} x^{5}  + 4a^{4} x^{3}  + 6a^{5} x 
\end{align}
where $a$ is some fixed element of $\ff_{17}$. Since $gcd(11,17^2-1) = 1$, this polynomial is a permutation polynomial\footnote{From \cite{MullenPanario}, it is known that a Dickson polynomial of degree $d$ over a finite field $\ff_q$ is a permutation polynomial if and only if $gcd(d,q^2-1) = 1$} over $\ff_{17}$ for all $a \in \ff_{17}$. It is first shown that $D_{11}(x,a)$ is indeed a permutation polynomial by verifying that the matrix $M_a$ for $D_{11}(x,a)$ is unimodular. 

The linear representation is constructed and the linear complexity of $D_{11}(x,a)$ is found to be $8$. The matrix $M_a$ has a determinant 
\[
\mbox{det} \ {M_a} = \frac{9 (a^2 - 5)  (a^6 + 2a^4 + 4a^2 - 5)  (a^8 - 6a^6 + 2a^4 - a^2 + 5)}{12  (a^4 - 4a^3 - 5a^2 - 3a - 7) (a^4 + 4a^3 - 5a^2 + 3a - 7) (a^8 + 3a^4 - 5a^2 - 7)}.
\]
It can be seen that the determinant is non-zero each $a \in \ff_{17}$ (as the numerator has no linear factors), and hence the inverse of $M_a$ exists for all $a$, and hence the function $D_{11}(x,a)$ is a permutation polynomial. 

The inverse of $D_{11}(x,a)$ is constructed as per the formula given in equation (\ref{inverforminonevar}) (and by using the fact that $a^{17} = a$).
\begin{align}
\begin{aligned}
D_{11}(x,a)^{-1} =& \frac{1}{d} \Bigg[ (4 a^{15} + 7 a^{13} + 2 a^{11} + 9 a^{9} + 6 a^{7} + 14 a^{5} + 10 a^{3} + 12 a) x^{13} \\
& \quad + (12 a^{16} + 14 a^{14} + 14 a^{10} + 4 a^{8} -  a^{4} + 9 a^{2}) x^{11} \\
& \quad + (3 a^{15} + 13 a^{13} + 7 a^{11} + 9 a^{9} + 2 a^{7} + 8 a^{5} + 15 a^{3} + 6 a) x^9 \\
& \quad + (a^{16} -  a^{14} + 12 a^{12} + a^{10} + 3 a^{8} + 6 a^{6} + 3 a^{4} + 4 a^{2}) x^7 \\
& \quad + (12 a^{15} + 10 a^{13} + 7 a^{11} + 13 a^{7} + 4 a^{5} + 11 a^{3} + 14 a) x^5 \\
& \quad + (11 a^{16} + 15 a^{14} + a^{12} + 8 a^{10} + 13 a^{8} + 10 a^{6} + 9 a^{4} + 5 a^{2} + 3) x^3 \\
& \quad + (7 a^{13} + 10 a^{11} + 2 a^{9} + a^{7} + 13 a^{5} + 8 a^{3} + 9 a) x \Bigg]
\end{aligned}
\label{eq:finvDP}
\end{align}
where 
\[
d = 9 (a^2 - 5) (a^6 + 2a^4 + 4a^2 - 5)(a^8 - 6a^6 + 2a^4 - a^2 + 5).
\]
The function was indeed validated to be the inverse of $f$ for each $a \in \ff_{17}$. As an illustration, when $a = 9 \in \ff_{17}$, the function $D_{11}(x,a)$ becomes 
\[
D_{11}(x,9) = x^{11} + 3x^9 +11x^7 + x^5 +13x^3 +14x
\]
and when $a$ is substituted in (\ref{eq:finvDP}) the inverse function is computed to be 
\[
D_{11}(x,9)^{-1} = 9x^{13} + 13x^{11} + 11x^9 + 12x^7 + 11x^5 + 11x^3 + 8x.
\]

\section{Representation of multivariate maps $F$ over $\ff^n$}
\label{sec:LR-mv}
In this section, we consider the problem of extending the linear representation of a function with one variable to maps $F : \fn \to \fn$ in $n$ variables. Such a map $F$ is defined by an $n$-tuple of polynomial functions $f_i(x_1,\dots,x_n)$
\[
F(x)=(f_1(x_1,\dots,x_n),\ldots,f_n(x_1,\ldots,x_n))^T
\]
where $f_i$ are polynomial functions $f_i:\fn \to \ff$. The Koopman operator of $F$ is the linear map $\mckf$, $\mckf*:\ffnz \to \ffnz$ defined on $\ff$-valued functions $\phi \in \ffnz$ in the following way
\[
\mckf(\phi)(x)=\phi(F(x_1,\ldots,x_n))\ \forall\ \phi \in \ffnz.  
\]
The coordinate functions $\chi_i$ $i = 1,\dots,n$ are defined as $\chi_i(x)=x_i$ for $x \in \fn$. It is to be noted that the space of functions $\ffnz$ is finite-dimensional and of dimension $q^n$ when the field is $\fq$.

\subsection{Construction of an $\mckf$-invariant subspace of $\ffnz$}
In the case of a univariate map, we constructed the cyclic subspace $S(f,\chi)$ and defined the linear representation by restricting $\mckf$ to this subspace. As a natural extension to the multi-variable case, we construct the smallest $\mckf$-invariant subspace $W$ containing all the coordinate functions. Algorithm \ref{alg:CycSub} constructs this invariant subspace.
\begin{algorithm}[H]
\begin{algorithmic}[1]
\caption{Construction of $W$: the $\mckf$-invariant subspace containing all $\chi_i(x)$}
\label{alg:CycSub}

\Procedure{$\mckf$-Invariant subspace containing $\chi_i(x)$}{}
\State \textbf{Outputs}
\begin{itemize}
    \item[] $W$ - the $\mckf$-invariant subspace which contains the coordinate functions. 
    \item[] $\mathcal{B}$ - a basis for the subspace $W$.
\end{itemize}
\State Compute the cyclic subspace  
\[
Z(\chi_1; \mckf) = \langle \chi_1,  \mckf\chi_1,\dots, (\mckf)^{l_1-1} \chi_1 \rangle.
\]
\State Choose a basis $\mathcal{B} = \{\chi_1,\mckf \chi_1,\dots,(\mckf)^{l_1-1}\chi_1 \}$.
\If{$\chi_2,\chi_3,\dots,\chi_n \in \mbox{Span}\{\mathcal{B}\}$} 
\State $W \gets \mbox{Span}\{\mathcal{B}\}$.
\State \textbf{halt}
\Else
\State Find the smallest $i$ such that $\chi_i \notin \mbox{span}\{\mathcal{B}\}$.
\State Compute the smallest $l_i$ such that 
\[
(\mckf)^{l_i} \chi_i \in \mbox{span}\{\mathcal{B}\} \cup \langle \chi_i,\mckf \chi_i,\dots, (\mckf)^{l_i-1}\chi_i \rangle.
\]
\State $V_i = \{\chi_i,\mckf \chi_i,\dots, (\mckf)^{l_i-1} \chi_i \}$.
\State Append the set $V_i$ to $\mathcal{B}$.
\State \textbf{go to} 5.
\EndIf
\EndProcedure
\end{algorithmic}
\end{algorithm}
Once this $\mckf$ invariant subspace $W$ is constructed, compute the matrix representation of $\mckf|W$ with respect to the chosen basis $\mcb$ of $W$. Let the dimension of this invariant subspace be $N$ and a basis of $W$ be denoted as
\beq\label{basisofW}
\mcb=\{\psi_1,\psi_2,\ldots,\psi_N\}.
\eeq
Let the matrix representation of $K_F=\mckf|W$ in $\mcb$ be denoted as $M$. (The notation for matrix representation is explained in (\ref{MatrixRepNotation})). Analogous to the univariate case, the dimension $N$ of the space $W$ is defined as the \emph{linear complexity} of the map $F$

\subsection{Linear representation of $F$ over $\fn$ defined by $\mckf|W$}
Since each of the coordinate functions $\chi_i(x)$ are in the subspace spanned by $\mcb$, every coordinate function has unique representation
\[
\chi_i(x)=v_i^T\hat{\psi}
\]
with a unique $v_i$ in $\ff^N$, where $\hat{\psi}$ is the $N$-tuple of basis functions $\psi_i(x)$
\[
\hat{\psi} = [\psi_1,\psi_2, \dots, \psi_N]^T.
\]
Construct the matrix $V$ as
\beq\label{Vmatrix}
V=
\left[
\begin{array}{c}
v_1^T\\v_2^T\\\vdots\\v_N^T
\end{array}
\right].
\eeq
Then we have $\hat{\chi}=V\hat{\psi}$ where $\hat{\chi}$ is the $n$-tuple of coordinate functions 
\[
\hat{\chi}=[\chi_1,\chi_2,\ldots,\chi_n]^T.
\] 
Using this representation of $\hat{\chi}$ we get 
\beq\label{RepofF}
\begin{array}{lcl}
F(x) & = & \hat{\mckf(\chi_i)}(x)\\
 & = & VM\hat{\psi}(x)
\end{array}
\eeq
which is called the \emph{linear representation} of the map $F$. Such a formula for a map $F$ should be valuable for various computations. We next show that this formula also carries over to the representation of $F^{-1}$. 

\subsection{Representation of $F^{-1}$}
First, we establish a condition for the invertibility of $F$ in terms of the representation (\ref{RepofF}).

\begin{lemma}
$F$ is invertible if and only if $M$ is non-singular.
\end{lemma}

\begin{proof}
Any function $\phi:\ff^n\rightarrow\ff$ satisfying $\mckf(\phi)=0$ satisfies $\phi(F(x))=0$ for all $x$. Hence when $F$ is invertible, $\phi(x)=0$ for all $x$ which proves $\phi$ is the zero function; hence $\mckf$ is one-to-one. Hence the restriction $\mckf|W$ is also one-to-one which implies $M$ is nonsingular. This proves the necessity.

Conversely, if $F$ is not invertible, then there exists distinct $x_1, x_2 \in \fn$ such that $F(x_1) = F(x_2)$. We know that by construction, all co-ordinate functions $\chi_i$ belong to $W$. Given a $\phi \in W$ and we have 
\begin{align*}
F(x_1) = F(x_2) \Leftrightarrow &\ \mckf \phi(x_1) = \mckf \phi(x_2) \\
 \Rightarrow &\ \psi(x_1) = \psi(x_2) \, \forall\, \psi\in\, \mckf(W)\\
  \Rightarrow &\ \chi_i\, \notin\, \mckf(W) \, \mbox{for some}\ \ i \\
   \Rightarrow &\  \mckf|W \mbox{ is not one-to-one on }\, W.
\end{align*}
This proves the sufficiency. Since $M$ is the matrix representation of $\mckf|W$, the matrix statement follows. 

\end{proof}

Next, we develop the representation for the compositional inverse $G$ of $F$. 

\begin{theorem}
If $F$ is invertible and $G=F^{-1}$ then $G$ has the representation
\beq\label{RepofFinv}
G(x)=VM^{-1}\hat{\psi}(x)
\eeq
where $V$ and $M$ are defined uniquely by the representation (\ref{RepofF}).
\end{theorem}

\begin{proof}
The representation of $\chi_i$ in $\mcb$ is given by $(v_i)^T\hat{\psi}$ where $\hat{\psi}=(\psi_1,\psi_2,\ldots,\psi_N)^T$ where $\{\psi_i,i=1,\ldots,N\}$ is the basis $\mcb$. Since $M$ is nonsingular when $F$ is invertible by previous lemma, it follows that
\[
(G^*)(\chi_i)=(\mckf)^{-1}(\chi_i)=(v_i)^T(F*)^{-1}\hat{\psi}=(v_i)^T(M)^{-1}\hat{\psi},
\]
which proves the form of coefficients of $(G^*)(\chi_i)$ where the matrix $V$ is defined by (\ref{Vmatrix}).
\end{proof}

\subsection{Illustrative example}
Consider the Feedback Shift Register (FSR) defined as a map $F : \ff_2^3 \to \ff_2^3$ given by the following equation
\[
F(x_1,x_2,x_3)= 
\begin{bmatrix} 
x_2 \\ x_3 \\ x_1 + x_2x_3
\end{bmatrix}.
\]
The $\mckf$-invariant subspace is constructed starting from $\chi_1(x) = x_1$ using Algorithm \ref{alg:CycSub} as follows
\begin{align*}
&x_1 \to x_2 \to x_3 \to x_1 + x_2 x_3 \to x_2 + x_1x_3 + x_2x_3 \to x_3 + x_1x_2 + x_1x_3 \to x_1 + x_1x_2 + x_1x_3.
\end{align*}
The last function is linearly dependent on the previous function as follows
\[
x_1 + x_1x_2 + x_1x_3 = x_1 + x_2 + x_3 + (x_2+x_1x_3+x_2x_3) + (x_3+x_1x_2+x_1x_3).
\]
A basis of the cyclic space can be chosen as 
\[
\mcb = \{ x_1, x_2, x_3, x_1 + x_2 x_3, x_2 + x_1x_3+ x_2x_3, x_3 + x_1x_2 + x_1x_3 \}.
\]
The matrix representation of the Koopman operator restricted to this subspace is obtained as
\[
M = \begin{bmatrix} 
0 & 1 & 0 & 0 & 0 & 0 \\
0 & 0 & 1 & 0 & 0 & 0 \\
0 & 0 & 0 & 1 & 0 & 0 \\
0 & 0 & 0 & 0 & 1 & 0 \\
0 & 0 & 0 & 0 & 0 & 1 \\
1 & 1 & 1 & 0 & 1 & 1
\end{bmatrix}.
\]
$M$ is verified to be nonsingular and hence $F$ is invertible. 
Each of the coordinate functions is represented in $\mcb$ as 
\begin{align*}
\chi_1  = e_1^T\hat{\psi}, \quad \quad \chi_2  =  e_2^T\hat{\psi}, \quad \quad \chi_3  =  e_3^T\hat{\psi}
\end{align*}
where $e_i$ are first three Cartesian vectors of $\ff_2^6$ and $\hat{\psi}$ is the vector of basis functions. $F$ has a linear representation
\[
F(x)=
\begin{bmatrix}
e_1\\e_2\\e_3
\end{bmatrix}
M\hat{\psi}(x)=
V M \hat{\psi}
\]
Since $M$ is nonsingular, $F$ has an inverse and can be computed using linear representation. In terms of $M^{-1}$ the representation of $F^{-1}$ is computed as
\[
\begin{array}{lcl}
F^{-1}(x) & = &
V M^{-1}\hat{\psi}(x) = \begin{bmatrix}
x_3+x_1x_2\\x_1\\x_2
\end{bmatrix}.
\end{array}
\]
It is easily verified that 
\[
F\bigg( \begin{bmatrix} x_3 + x_1 x_2 \\ x_1 \\ x_2 \end{bmatrix} \bigg) = \begin{bmatrix} x_1 \\ x_2 \\ x_3 \end{bmatrix}.
\]
\subsection{Solutions to a system of multivariate difference equations}
A system of multivariate difference equations over $\ff$ is defined as 
\[
\begin{bmatrix}
x_1(k+1) \\ x_2(k+1) \\ \vdots \\ x_n(k+1) 
\end{bmatrix} = \begin{bmatrix}
f_1(x_1(k),\dots,x_n(k)) \\ f_2(x_1(k),\dots,x_n(k)) \\ \vdots \\ f_n(x_1(k),\dots,x_n(k)) 
\end{bmatrix}
\]
where $x_i \in \ff$. Such systems occur commonly in dynamical systems theory and are known as the state space representation of a dynamical system. Solutions to such systems involve computing the value of the state $x(k):= [x_1(k),\dots,x_n(k)]^T$ starting from initial condition $x(0)$. Since $f_i$ are non-linear, computing the solutions is a hard problem, and the linear representation of $F = [f_1,\dots,f_n]^T$ gives a linear algebraic computational methodology for the same. The solution $x(k)$ for any $k$ starting from an initial condition $x(0)$ is given by
\[
x(k) = V M^k \hat{\psi}(x(0)).
\]
Further details about the dynamic nature of the map $F$ using the linear representation is reported in the work \cite{RamSule}. 

\section{Linear representation of the group generated by multiple invertible maps}
\label{sec:LR-FG}

Consider a finite set of non-singular maps $F_1,F_2,\dots,F_N$ 
\[
F_i : \ff^n \to \ff^n.
\]
A finite group, $G_F$, can be generated from $F_i$ using composition as the group operation. In this section, we devise a procedure to compute the linear representation of the group $G_F$ over $\ff$. The motivation for such a linear representation is as follows. If such a representation is developed, the non-linear operations $F_i$ over $\ff^n$ and their compositions can be studied in linear algebraic computations. In order to extend the notion of linear representation developed in this work, we construct the smallest subspace of $\ffnz$, which contains all the coordinate functions $\chi_i$ and is invariant to all the operators $\mckf_i$. 

\subsection*{Computing an invariant subspace for multiple operators over $\ffnz$}
We devise Algorithm \ref{alg:GroupW1} to construct the subspace $W \subset \ffnz$ satisfying the following properties
\begin{itemize}
    \item $W$ contains the function $\chi_i$.
    \item $W$ is invariant to each of the linear operators $\mckf_i$ acting on $\ffnz$.
    \item $W$ is the smallest such subspace satisfying the above two properties.
\end{itemize}

\begin{algorithm}[h]
\begin{algorithmic}[1]
\caption{Construction of $W$: the smallest $\mckf_i$-invariant subspace containing all $\chi_i(x)$ for maps $F_1,\dots, F_N$}
\label{alg:GroupW1}
\Procedure{Invariant subspace containing $\chi_i(x)$}{}
\State \textbf{Outputs}
\begin{itemize}
    \item[] $W$ - the smallest subspace containing all the coordinate functions and invariant to each $\mckf_i$.
\end{itemize}
\State \textbf{Inputs}
\begin{itemize}
    \item[] $W_i$ - the smallest $\mckf_i$-invariant subspace containing all the coordinate functions using Algorithm \ref{alg:CycSub}.
\end{itemize}
\State Initiate $W = \sum W_i$.
\For{$i \in \{1,2,\dots,N \}$}
    \State Compute $W_i$ - the smallest $\mckf_i$-invariant subspace containing $W$.
\EndFor
\State $\bar{W} = \sum W_i$
\If{ $W$ equals $\bar{W}$}
    \State $W$ is the required subspace invariant to each $F_i$ and contains all coordinate functions.
    \State Halt.
\Else
    \State Replace $W$ by $\bar{W}$. 
    \State Goto step 5.
\EndIf

\EndProcedure
\end{algorithmic}
\end{algorithm}

Once this subspace $W$ is constructed, and a basis $\mcb$ for $W$ is chosen, the restriction of $\mckf_i|W$ can be written as a matrix $M_i$. Each of these maps $F_i$ can then be associated with a matrix $M_i$ with the basis $\mcb$. If the dimension of $W$ is $N_G$, all the matrices $M_i$ are of dimension $N_G$ over $\ff$.

\subsection{Representation of the group $G_F$}

Given the matrix representations $M_i$ for each of $\mckf_i|W$ using the basis $\mcb$, we look for a representation $\sigma: G_F \to GL(N_G,\ff)$ satisfying the following. The evaluation of $\sigma$ on each $F_i$ gives the matrix $M_i$ and given a map $g \in G_F$, $\sigma(g) = M_g$,  where $M_g$ satisfies $M_g = \mckf_g|W$ for the chosen basis $\mcb$.

Considering $g = F_iF_j := F_i(F_j(x))$
\begin{align}
    \label{eq:GrpRep}
\begin{aligned}
    \sigma(g) &= \sigma(F_i F_j) \\
            & = \mckf_{F_1 F_j} | W \\
            &= M_i M_j
\end{aligned}
\end{align}

Similarly for $g = F_i^{-1}$
\begin{align*}
    \sigma(g) &= \sigma(F_i^{-1}) \\
            &= \mckf_{F_i^{-1}}|W \\
            &= M_i^{-1}
\end{align*}
The equation (\ref{eq:GrpRep}) gives the representation $M \in GL(N_G,\ff)$ for any word $g \in G_F$ in terms of products of the matrices $M_i$. For example, consider $g = F_1F_2^2F_3F_1$, then
\[
\sigma(g) = M_1 M_2^2 M_3 M_1
\]
The above relation shows that $\sigma$ is a homomorphism from $G$ to $GL(N_G,\ff)$. The dimension $N_G$ of the subspace $W$, which is the smallest subspace invariant to all maps $\mckf_i$ and containing the coordinate functions $\chi_i$ is defined as the \textit{linear complexity} of the group $G_F$.

\section{Conclusion}
The paper primarily addresses the problem of linear representation, invertibility, and construction of the compositional inverse for non-linear maps over finite fields. Though there is vast literature available for invertibility of polynomials and construction of inverses of permutation polynomials over $\ff$, this paper explores a completely new approach using the Koopman operator defined by the iterates of the map. This helps define the linear representation of non-linear maps, which translates non-linear compositions of the map to matrix multiplications. This linear representation naturally defines a notion of linear complexity for non-linear maps, which can be viewed as a measure of computational complexity associated with computations involving such maps. The framework of linear representation is then extended to parameter dependent maps over $\ff$, and the conditions on parametric invertibility of such maps are established, leading to a construction of the parametric inverse map (under composition). It is shown that the framework can be extended to multivariate maps over $\ff^n$, and the conditions are established for invertibility of such maps, and the inverse is constructed using the linear representation. Further, the problem of linear representation of the group generated by a finite set of permutation maps over $\ff^n$ under composition is also solved by extending the theory of linear representation of a single map. This leads to the notion of complexity of a group of permutation maps under composition. 

\subsection*{Acknowledgements}
The first author would like to thank the Department of Electrical Engineering, Indian Institute of Technology - Bombay, as the work was done in full during his tenure as a Institue Post-Doctoral Fellow. The authors would also like to thank the reviewers for their suggestions in the proofs of Lemma 1, Proposition 1 and Lemma 3.  

\bibliographystyle{elsarticle-num-names}
\bibliography{mybib.bib}

\begin{thebibliography}{30}
\expandafter\ifx\csname natexlab\endcsname\relax\def\natexlab#1{#1}\fi
\providecommand{\url}[1]{\texttt{#1}}
\providecommand{\href}[2]{#2}
\providecommand{\path}[1]{#1}
\providecommand{\DOIprefix}{doi:}
\providecommand{\ArXivprefix}{arXiv:}
\providecommand{\URLprefix}{URL: }
\providecommand{\Pubmedprefix}{pmid:}
\providecommand{\doi}[1]{\href{http://dx.doi.org/#1}{\path{#1}}}
\providecommand{\Pubmed}[1]{\href{pmid:#1}{\path{#1}}}
\providecommand{\bibinfo}[2]{#2}
\ifx\xfnm\relax \def\xfnm[#1]{\unskip,\space#1}\fi
\bibitem[{Gill(1962)}]{Gill2}
\bibinfo{author}{A.~R. Gill}, \bibinfo{title}{Introduction to the Theory of Finite State Machines}, \bibinfo{publisher}{McGraw-Hill Inc.,US}, \bibinfo{year}{1962}.
\bibitem[{Mortveit and Reidys(2008)}]{SeqDyn}
\bibinfo{author}{H.~S. Mortveit}, \bibinfo{author}{C.~M. Reidys}, \bibinfo{title}{An Introduction to Sequential Dynamical Systems}, \bibinfo{publisher}{Springer}, \bibinfo{year}{2008}.
\bibitem[{Golomb(1982)}]{GolombFSR}
\bibinfo{author}{S.~W. Golomb}, \bibinfo{title}{Shift Register Sequences}, \bibinfo{publisher}{Aegean Park Press}, \bibinfo{year}{1982}.
\bibitem[{Goresky and Klapper(2012)}]{Goresky}
\bibinfo{author}{M.~Goresky}, \bibinfo{author}{A.~Klapper}, \bibinfo{title}{Algebraic Shift Register Sequences}, \bibinfo{publisher}{Cambridge University Press}, \bibinfo{year}{2012}.
\bibitem[{Rueppel(1986)}]{Rueppel}
\bibinfo{author}{R.~A. Rueppel}, \bibinfo{title}{Analysis and Design of Stream Ciphers}, \bibinfo{publisher}{Springer Berlin Heidelberg}, \bibinfo{year}{1986}.
\bibitem[{Jong(2002)}]{HDJ}
\bibinfo{author}{H.~D. Jong},
\newblock \bibinfo{title}{Modeling and simulation of genetic regulatory systems : A literature review},
\newblock \bibinfo{journal}{Journal of Computational Biology} \bibinfo{volume}{9} (\bibinfo{year}{2002}) \bibinfo{pages}{67--103}.
\bibitem[{Kauffman(1969)}]{Kauffman}
\bibinfo{author}{S.~A. Kauffman},
\newblock \bibinfo{title}{Metabolic stability and epigenesis in randomly constructed genetic nets},
\newblock \bibinfo{journal}{Journal of Theoretical Biology} \bibinfo{volume}{22} (\bibinfo{year}{1969}) \bibinfo{pages}{437--467}. \DOIprefix\doi{10.1016/0022-5193(69)90015-0}.
\bibitem[{Oishi and Klavins(2014)}]{KO_EK}
\bibinfo{author}{K.~Oishi}, \bibinfo{author}{E.~Klavins},
\newblock \bibinfo{title}{Framework for engineering finite state machines in gene regulatory networks},
\newblock \bibinfo{journal}{ACS Synthetic Biology} \bibinfo{volume}{3} (\bibinfo{year}{2014}) \bibinfo{pages}{652--665}. \DOIprefix\doi{10.1021/sb4001799}.
\bibitem[{Thomas(1973)}]{RThomas}
\bibinfo{author}{R.~Thomas},
\newblock \bibinfo{title}{Boolean formalization of genetic control circuits},
\newblock \bibinfo{journal}{Journal of Theoretical Biology} \bibinfo{volume}{42} (\bibinfo{year}{December 1973}) \bibinfo{pages}{563--585}. \DOIprefix\doi{10.1016/0022-5193(73)90247-6}.
\bibitem[{Mullen and Panario(2013)}]{MullenPanario}
\bibinfo{author}{G.~L. Mullen}, \bibinfo{author}{D.~Panario}, \bibinfo{title}{Handbook of finite fields}, \bibinfo{publisher}{CRC Press, Taylor and Francis Group}, \bibinfo{year}{2013}.
\bibitem[{A.~Cafure and Waissbein(2006)}]{Cafureetal}
\bibinfo{author}{G.~M. A.~Cafure}, \bibinfo{author}{A.~Waissbein},
\newblock \bibinfo{title}{Inverting bijective polynomial maps over finite fields},
\newblock \bibinfo{journal}{IEEE Information Theory Workshop - ITW '06, Punta del Este}  (\bibinfo{year}{2006}).
\bibitem[{Coulter and Henderson(2002)}]{CH_2002}
\bibinfo{author}{R.~S. Coulter}, \bibinfo{author}{M.~Henderson},
\newblock \bibinfo{title}{The compositional inverse of a class of permutation polynomials over a finite field},
\newblock \bibinfo{journal}{Bulletin of the Australian Mathematical Society} \bibinfo{volume}{65} (\bibinfo{year}{2002}) \bibinfo{pages}{521–526}. \DOIprefix\doi{10.1017/S0004972700020578}.
\bibitem[{Tuxanidy and Wang(2014)}]{TuxanidyWang}
\bibinfo{author}{A.~Tuxanidy}, \bibinfo{author}{Q.~Wang},
\newblock \bibinfo{title}{On the inverses of some classes of permutations of finite fields},
\newblock \bibinfo{journal}{Finite Fields and Their Applications} \bibinfo{volume}{28} (\bibinfo{year}{2014}) \bibinfo{pages}{244--281}. \DOIprefix\doi{10.1016/j.ffa.2014.02.006}.
\bibitem[{Zheng et~al.(2020)Zheng, Wang, and Wei}]{ZhengWangWei}
\bibinfo{author}{Y.~Zheng}, \bibinfo{author}{Q.~Wang}, \bibinfo{author}{W.~Wei},
\newblock \bibinfo{title}{On inverses of permutation polynomials of small degree over finite fields},
\newblock \bibinfo{journal}{IEEE Transactions on Information Theory} \bibinfo{volume}{66} (\bibinfo{year}{2020}) \bibinfo{pages}{914--922}. \DOIprefix\doi{10.1109/TIT.2019.2939113}.
\bibitem[{Laigle-Chapuy(2007)}]{YLC}
\bibinfo{author}{Y.~Laigle-Chapuy},
\newblock \bibinfo{title}{Permutation polynomials and applications to coding theory},
\newblock \bibinfo{journal}{Finite Fields and Their Applications} \bibinfo{volume}{13} (\bibinfo{year}{2007}) \bibinfo{pages}{58--70}. \DOIprefix\doi{10.1089/10665270252833208}.
\bibitem[{Levine and Brawley(1977)}]{JL_JB}
\bibinfo{author}{J.~Levine}, \bibinfo{author}{J.~V. Brawley},
\newblock \bibinfo{title}{Some cryptographic applications of permutation polynomials},
\newblock \bibinfo{journal}{Cryptologia} \bibinfo{volume}{1} (\bibinfo{year}{1977}) \bibinfo{pages}{76--92}. \DOIprefix\doi{10.1080/0161-117791832814}.
\bibitem[{Çeşmelioğlu et~al.(2008)Çeşmelioğlu, Meidl, and Topuzoğlu}]{Ayca}
\bibinfo{author}{A.~Çeşmelioğlu}, \bibinfo{author}{W.~Meidl}, \bibinfo{author}{A.~Topuzoğlu},
\newblock \bibinfo{title}{On the cycle structure of permutation polynomials},
\newblock \bibinfo{journal}{Finite Fields and Their Applications} \bibinfo{volume}{14} (\bibinfo{year}{2008}) \bibinfo{pages}{593--614}. \DOIprefix\doi{10.1016/j.ffa.2007.08.003}.
\bibitem[{Lidl and Mullen(1991)}]{LidlMullen_Dickson}
\bibinfo{author}{R.~Lidl}, \bibinfo{author}{G.~Mullen},
\newblock \bibinfo{title}{Cycle structure of dickson permutation polynomials},
\newblock \bibinfo{journal}{Mathematical Journal of Okayama Univeristy} \bibinfo{volume}{33} (\bibinfo{year}{1991}) \bibinfo{pages}{1--11}. \DOIprefix\doi{10.1016/j.ffa.2007.08.003}.
\bibitem[{Mullen et~al.(2016)Mullen, Muratovi{\'c}-Ribi{\'c}, and Wang}]{Mullen2016}
\bibinfo{author}{G.~L. Mullen}, \bibinfo{author}{A.~Muratovi{\'c}-Ribi{\'c}}, \bibinfo{author}{Q.~Wang},
\newblock \bibinfo{title}{On coefficients of powers of polynomials and their compositionsover finite fields},
\newblock \bibinfo{journal}{Contemporary Developments in Finite Fields and Applications}  (\bibinfo{year}{2016}) \bibinfo{pages}{270--281}. \DOIprefix\doi{10.1142/9789814719261_0016}.
\bibitem[{R.~Lidl and Tumwald(1993)}]{LidlBook}
\bibinfo{author}{G.~M. R.~Lidl}, \bibinfo{author}{G.~Tumwald}, \bibinfo{title}{Dickson Polynomials}, \bibinfo{publisher}{Pitman Monographs and Surveys in Pure and Applied Mathematics}, \bibinfo{year}{1993}.
\bibitem[{Wu and Liu(2013)}]{BW_ZL}
\bibinfo{author}{B.~Wu}, \bibinfo{author}{Z.~Liu},
\newblock \bibinfo{title}{Linearized polynomials over finite fields revisited},
\newblock \bibinfo{journal}{Finite Fields and their Applications} \bibinfo{volume}{22} (\bibinfo{year}{2013}) \bibinfo{pages}{79--100}. \DOIprefix\doi{10.1016/j.ffa.2013.03.003}.
\bibitem[{Berlekamp(2015)}]{Berlekamp}
\bibinfo{author}{E.~R. Berlekamp}, \bibinfo{title}{Algebraic Coding Theory (Revised Edition)}, \bibinfo{publisher}{World Scientific Publishing Company Private Ltd}, \bibinfo{year}{2015}.
\bibitem[{Massey(1969)}]{Massey}
\bibinfo{author}{J.~Massey},
\newblock \bibinfo{title}{Shift-register synthesis and bch decoding},
\newblock \bibinfo{journal}{IEEE Transactions on Information Theory} \bibinfo{volume}{15} (\bibinfo{year}{1969}) \bibinfo{pages}{122--127}. \DOIprefix\doi{10.1109/TIT.1969.1054260}.
\bibitem[{Gill(1966)}]{Gill1}
\bibinfo{author}{A.~R. Gill}, \bibinfo{title}{Linear Sequential Circuits- Analysis, Synthesis and Applications}, \bibinfo{publisher}{McGraw-Hill Book Company}, \bibinfo{year}{1966}.
\bibitem[{Harrison(1969)}]{Harrison}
\bibinfo{author}{M.~A. Harrison}, \bibinfo{title}{Lectures on Linear Sequential Machines}, \bibinfo{publisher}{Academic Press, New York and London}, \bibinfo{year}{1969}.
\bibitem[{Anantharaman and Sule(2021)}]{RamSule}
\bibinfo{author}{R.~Anantharaman}, \bibinfo{author}{V.~Sule},
\newblock \bibinfo{title}{Koopman operator approach for computing structure of solutions and observability of nonlinear dynamical systems over finite fields},
\newblock \bibinfo{journal}{Mathematics of Control Signals and Systems} \bibinfo{volume}{33} (\bibinfo{year}{2021}) \bibinfo{pages}{331--358}. \DOIprefix\doi{10.1007/s00498-021-00286-y}.
\bibitem[{Koopman(1931)}]{Koopman}
\bibinfo{author}{B.~O. Koopman},
\newblock \bibinfo{title}{Hamiltonian systems and transformation in hilbert space},
\newblock \bibinfo{journal}{Proceedings of the National Academy of Sciences of the United States of America} \bibinfo{volume}{17(5)} (\bibinfo{year}{1931}) \bibinfo{pages}{315--318}.
\bibitem[{Williams et~al.(2015)Williams, Kevrekidis, and Rowley}]{Mwilliams}
\bibinfo{author}{M.~O. Williams}, \bibinfo{author}{I.~G. Kevrekidis}, \bibinfo{author}{C.~W. Rowley},
\newblock \bibinfo{title}{A data driven approximation of the koopman operator: Extending dynamic mode decomposition},
\newblock \bibinfo{journal}{Journal of Nonlinear Science} \bibinfo{volume}{25(6)} (\bibinfo{year}{2015}) \bibinfo{pages}{1307--1346}. \DOIprefix\doi{10.1007/s00332-015-9258-5}.
\bibitem[{Mauroy et~al.(2020)Mauroy, Mezi\'{c}, and Susuki}]{Mauroy}
\bibinfo{author}{A.~Mauroy}, \bibinfo{author}{I.~Mezi\'{c}}, \bibinfo{author}{Y.~Susuki}, \bibinfo{title}{The Koopman Operator in Systems and Control: Concepts, Methodologies, and Applications}, \bibinfo{publisher}{Springer}, \bibinfo{year}{2020}.
\bibitem[{Muratović-Ribić(2007)}]{AMR}
\bibinfo{author}{A.~Muratović-Ribić},
\newblock \bibinfo{title}{A note on the coefficients of inverse polynomials},
\newblock \bibinfo{journal}{Finite Fields and Their Applications} \bibinfo{volume}{13(4)} (\bibinfo{year}{2007}) \bibinfo{pages}{977--980}. \DOIprefix\doi{10.1016/j.ffa.2006.11.003}.

\end{thebibliography}
\end{document}